\documentclass[a4paper,12pt]{article}
\usepackage{amssymb,amsmath,bm}
\usepackage{graphicx}
\usepackage{enumerate}
\usepackage{mathrsfs}
\usepackage{amsthm}

\title{The Relationship between Consumer Theories with and without Utility Maximization}
\author{Yuhki Hosoya\thanks{E-mail: ukki(at)gs.econ.keio.ac.jp,\ ORCID ID:0000-0002-8581-4518}~{}\thanks{Faculty of Economics, Chuo University. 742-1, Higashinakano, Hachioji, Tokyo, 192-0393 Japan.}}

\makeatletter
\renewenvironment{proof}[1][\proofname]{\par
  \pushQED{\qed}%
  \normalfont \topsep6\p@\@plus6\p@\relax
  \trivlist
  \item\relax
  {\bfseries
  #1\@addpunct{.}}\hspace\labelsep\ignorespaces
}{%
  \popQED\endtrivlist\@endpefalse
}
\makeatother

\theoremstyle{definition}
\newtheorem{prop}{Proposition}
\newtheorem{thm}{Theorem}
\newtheorem{lem}{Lemma}
\newtheorem{cor}{Corollary}
\newtheorem{step}{Step}
\newtheorem{fact}{Fact}

\begin{document}
\maketitle

\begin{abstract}
To study the assumption that the utility maximization hypothesis implicitly adds to consumer theory, we consider a mathematical representation of pre-marginal revolution consumer theory based on subjective exchange ratios. We introduce two axioms on subjective exchange ratio, and show that both axioms hold if and only if consumer behavior is consistent with the utility maximization hypothesis. Moreover, we express the process for a consumer to find the transaction stopping point in terms of differential equations, and prove that the conditions for its stability are equal to the two axioms introduced in the above argument. Therefore, the consumer can find his/her transaction stopping point if and only if his/her behavior is consistent with the utility maximization hypothesis. In addition to these results, we discuss equivalence conditions for axioms to evaluate their mathematical strength, and methods for expressing the theory of subjective exchange ratios in terms of binary relations.

\vspace{12pt}

\vspace{12pt}
\end{abstract}

\section{Introduction}
In modern economic theory, economic agents are usually designed to follow the utility maximization hypothesis. This remains true even when the rationality of the agent is in doubt, and studies that try to find a class of utility functions that can explain irrational behavior that emerges from psychological experiments are often found in decision theory. A typical example of such a study is Schmeidler (1989), who constructed a theory of Choquet expected utility to explain Ellsberg's (1961) experiment.

However, economics has not been with the utility maximization hypothesis since its beginnings. This hypothesis was introduced by William Stanley Jevons, Carl Menger, and Leon Walras during the period known as the {\bf marginal revolution}. Thus, the following question naturally arises: With the addition of the utility maximization hypothesis, what have we assumed about economic agents?

To consider this question, we need to compare consumer theory before and after utility maximization was assumed. Therefore, let us look at consumer theory before the marginal revolution. According to Kawamata (1989), in theories of consumer behavior prior to the marginal revolution, the consumer was assumed to be able to view his/her {\bf subjective value} for $i$-th commodity according to their current state. In this case, the `appropriate' exchange ratio between $i$-th and $j$-th commodities the consumer believes would be expressed as a ratio of subjective values. Let us call this ratio the {\bf subjective exchange ratio}. On the other hand, the actual exchange ratio is expressed as a ratio of prices. If these ratios are different, then the consumer will not be satisfied with what he/she has and will continue to trade. As a result of the transaction, the subjective value and the subjective exchange ratio are updated. After repeating this process, the consumer will stop trading only when the subjective exchange ratio finally becomes equal to the price ratio.

Menger (1871) showed that if we replace the subjective value treated in this pre-marginal revolution consumer theory with the concept of {\bf marginal utility}, everything can be well explained. If the subjective value is given by marginal utility, then the subjective exchange ratio coincides with the usual {\bf marginal rate of substitution}. Moreover, from Lagrange's first-order condition, the point at which the transaction stops is precisely the point of utility maximization. Thus, the new theory modeled in terms of utility maximization can be positioned as an inheritance and development of the old theory. However, one can also consider the consumer who attaches subjective exchange ratios that cannot be regarded as marginal utility. Menger showed that the utility-maximizing consumer can also be described as the consumer with subjective exchange ratios, but the converse is not true. Thus, another view could be that Menger made stronger assumptions about the consumer. If so, what additional assumptions were made about consumer behavior by introducing utility maximization? In this paper, we consider this question. What is the relationship between new and old consumer theories?

Another aspect of the marginal revolution is that it introduced the convention of describing theories mathematically in economics. This means that the theory using subjective exchange ratios was not described mathematically. Therefore, this paper begins by mathematizing consumer theory based on subjective exchange ratios. We construct two mathematical theories: static and dynamic.

In both theories, the main consideration is the function that gives a vector of the subjective values of commodities when the consumption vector $x$ is on hand. In static theory, we consider the transaction stopping condition as the consumption vector in the budget hyperplane under which the subjective value vector is proportional to the price vector, and construct the corresponding demand function $f^g$. We impose two conditions for $g$. The first condition is called the {\bf weak weak axiom}, which has the same interpretation as the weak weak axiom of revealed preference introduced by Kihlstrom et al. (1976). The second condition is called {\bf Ville's axiom} and is a requirement for a certain type of rationality introduced by Ville (1946). We first show that under the locally Lipschitz assumption, $f^g$ can be written as a result of the utility maximization if and only if $g$ satisfies the above two axioms (Theorem \ref{THM1}). Thus, these two axioms are the additional conditions imposed on the subjective exchange ratio of the consumer by the marginal revolution.

Static theory describes only the transaction stopping point and does not discuss how the consumer finds that point. In contrast, dynamic theory discusses how the consumer finds the transaction stopping point. Since the consumer can only observe the current subjective value, he/she only knows the direction of improvement. In order to describe efforts to improve the consumer's state, we use some kind of differential equation. If the solution to the differential equation converges to the transaction stopping point, then we can interpret that the consumer can successfully find the transaction stopping point. We analyze this differential equation, and show that it is necessary and sufficient for stability to hold no matter what improvement method is used that both the weak weak axiom and Ville's axiom hold (Theorem \ref{THM2}). Thus, the consumer will always eventually find a transaction stopping point when they can be viewed as if they are acting according to the utility maximization hypothesis, but they may fail to find a transaction stopping point when they are not. This is another implication of the utility maximization hypothesis.

We have introduced two axioms, and the economic interpretations of these axioms are relatively natural. On the other hand, it is not easy to evaluate how mathematically strong these axioms are. We study this and find equivalent conditions for axioms that are easy to evaluate (Theorem \ref{THM3}). Moreover, the condition obtained there can be characterized by the properties of what was called the {\bf Antonelli matrix} in classical consumer theory (Proposition \ref{PROP1}). Using this result, we can obtain the following two conclusions. First, when there are only two commodities, any subjective value function satisfies Ville's axiom. When there are three or more commodities, however, the set of subjective value functions that satisfies Ville's axiom is nowhere dense under almost any appropriate topology (Corollary \ref{COR1}). Thus, when there are three or more commodities, Ville's axiom is, at least mathematically, very strong.

Finally, we discuss the relationship between the theory of subjective exchange ratios and the theory of binary relations. Although we formulated the primitives of consumers as a subjective value function, in consumer theory since the 20th century, it is usually assumed that the primitive mathematical element of the consumer is the binary relation that represents the consumer's preference. We mention a way to construct a binary relation from $g$ and provide a correspondence in a natural way (Theorem \ref{THM4}). The corresponding binary relation $\succsim^g$ is the unique continuous binary relation that satisfies a few additional requirements. This binary relation is a weak order if and only if $g$ satisfies Ville's axiom. Moreover, the subjective exchange ratios under $g$ and $h$ coincide at any point if and only if $\succsim^g=\succsim^h$, and the mapping $g\mapsto \succsim^g$ is continuous under natural topologies. Therefore, we can consider that this $\succsim^g$ is the consumer's preference relation that naturally corresponds to the subjective exchange ratio.

All of our results are obtained under the assumption that $g$ is locally Lipschitz. Although the most natural assumption for $g$ is continuity, for a consistent construction method of indifference curves, we could not do without locally Lipschitz assumption of $g$. On the other hand, all results of classical integrability theory have been obtained under continuous differentiability of $g$. Since any continuously differentiable function is locally Lipschitz, our assumption is weaker than that of related research.

This paper is organized as follows. Section 2 provides the notation and definitions that are necessary for the discussion in this paper. Subsection 3.1 describes the results of static theory, whereas Subsection 3.2 describes the results of the dynamic theory. Subsection 3.3 provides a series of discussions to evaluate the strength of axioms. Subsection 3.4 addresses the correspondence problem between subjective exchange ratios and binary relations. Section 4 discusses the relationship between the theory discussed in this paper and related studies. Section 5 is the conclusion. Since the proofs of results in this paper are lengthy, all proofs of results are placed in Section 6.

\section{Notation and Definitions}
Throughout this paper, we use the following notation: $\mathbb{R}^N_+=\{x\in\mathbb{R}^N|x_i\ge 0\mbox{ for all }i\in \{1,...,N\}\}$, and $\mathbb{R}^N_{++}=\{x\in \mathbb{R}^N|x_i>0\mbox{ for all }i\in \{1,...,N\}\}$. The former set is called the \textbf{nonnegative orthant} and the latter set is called the \textbf{positive orthant}. We write $x\ge y$ if $x-y\in \mathbb{R}^N_+$, $x\gneq y$ if $x\ge y$ and $x\neq y$, and $x\gg y$ if $x-y\in \mathbb{R}^N_{++}$. If $N=1$, then we omit $N$ and simply write $\mathbb{R}_+$ and $\mathbb{R}_{++}$.

Fix $n\ge 2$. Let $\Omega$ denote the consumption set. We assume that $\Omega=\mathbb{R}^n_{++}$ unless otherwise stated. A set $\succsim\subset \Omega^2$ is called a {\bf binary relation} on $\Omega$.

For a binary relation $\succsim$ on $\Omega$, we write $x\succsim y$ instead of $(x,y)\in \succsim$ and $x\not\succsim y$ instead of $(x,y)\notin \succsim$. Moreover, we write $x\succ y$ if $x\succsim y$ and $y\not\succsim x$, and $x\sim y$ if $x\succsim y$ and $y\succsim x$. Note that, $\succ$ and $\sim$ can be seen as binary relations.

For a binary relation $\succsim$ on $\Omega$, we say that it is
\begin{itemize}
\item {\bf complete} if, for every $(x,y)\in \Omega^2$, either $x\succsim y$ or $y\succsim x$,

\item {\bf transitive} if $x\succsim y$ and $y\succsim z$ imply $x\succsim z$,

\item {\bf continuous} if $\succsim$ is relatively closed in $\Omega^2$,

\item {\bf locally non-satiated} if for any $x\in \Omega$ and any neighborhood $U$ of $x$, there exists $y\in U$ such that $y\succ x$,

\item {\bf monotone} if $x\gg y$ implies $x\succ y$, and

\item {\bf weakly convex} if $y\succsim x$ and $t\in [0,1]$ imply $(1-t)x+ty\succsim x$.
\end{itemize}

A binary relation $\succsim$ on $\Omega$ is called a {\bf weak order} if it is complete and transitive. It can be easily shown that if $\succsim$ is a transitive binary relation on $\Omega$, then $\succ$ is also transitive.

Suppose that $\succsim$ is a weak order on $\Omega$. If there exists a function $u:\Omega\to \mathbb{R}$ such that
\[x\succsim y\Leftrightarrow u(x)\ge u(y),\]
then we say that $u$ {\bf represents} $\succsim$, or $u$ is a {\bf utility function} of $\succsim$. It is known that for a binary relation $\succsim$ on $\Omega$, it has a continuous utility function if and only if it is a continuous weak order.\footnote{See Debreu (1954).}

Let $\succsim$ be a binary relation on $\Omega$. For each $(p,m)\in (\mathbb{R}^n_+\setminus \{0\})\times \mathbb{R}_{++}$, we define
\[\Delta(p,m)=\{x\in \Omega|p\cdot x\le m\},\]
\[f^{\succsim}(p,m)=\{x\in \Delta(p,m)|y\not\succ x\mbox{ for all }y\in \Delta(p,m)\}.\]
We call the set-valued function $f^{\succsim}$ the {\bf demand relation} of $\succsim$, and if it is single-valued, then we call $f^{\succsim}$ the {\bf demand function} of $\succsim$. If $u$ represents $\succsim$, then $f^u$ denotes $f^{\succsim}$.

Note that, we do not assume that $f^{\succsim}(p,m)\neq \emptyset$ for all $(p,m)\in (\mathbb{R}^n_+\setminus\{0\})\times \mathbb{R}_{++}$. For a set-valued function $f:(\mathbb{R}^n_+\setminus \{0\})\times \mathbb{R}_{++}\twoheadrightarrow \Omega$, we call the set of $(p,m)$ in which $f(p,m)$ is nonempty the {\bf domain} of $f$.

\section{Main Results}
\subsection{Static Theory: Representation by Utility Maximization}
As noted in the introduction, the consumer theory before the marginal revolution was not formulated mathematically. This is inconvenient for us, and thus we must first reformulate this theory mathematically. There are two possible formulations, static and dynamic. In this subsection, we treat the static theory.

Let a locally Lipschitz vector field $g:\Omega\to \mathbb{R}^n_+\setminus \{0\}$ be given. For a given consumption vector $x\in \Omega$, $g_i(x)$ is interpreted as representing the consumer's {\bf subjective value} of $i$-th commodity when he/she has the consumption vector $x$, and its relative ratio $g_i(x)/g_j(x)$ represents his/her {\bf subjective exchange ratio} between the $i$-th and $j$-th commodities. If the consumer has the consumption vector $x$, then he/she believes that the ``appropriate'' exchange ratio of the $i$-th and $j$-th commodities for him/her is $g_i(x):g_j(x)$. On the other hand, given the market price $p$, the ``actual'' exchange ratio of the $i$-th and $j$-th commodities he/she faces is $p_i:p_j$. If these do not coincide, then he/she will not be satisfied with his/her current consumption vector $x$, and will attempt to continue the transaction.

We continue to consider the consumer's choice. Because of the interpretation of the subjective exchange ratio, if $g(x)\cdot v>0$, then the consumer that has the consumption vector $x$ thinks that $x+tv$ is more desirable than $x$ whenever $t>0$ is sufficiently small. Therefore, if the consumer's income is $m>0$ and $p\cdot (x+tv)\le m$ for a small $t>0$, then this consumer rejects this $x$. For the absence of such a vector $v$, it is necessary and sufficient that the following two requirements hold: 1) $p\cdot x=m$, and 2) $g(x)$ is proportional to $p$. The consumer agrees to terminate the transaction if and only if these two conditions are met. Under this consideration, we define a set-valued function $f^g$ that describes this requirement:
\[f^g(p,m)=\{x\in \Omega|p\cdot x=m,\ \exists \lambda>0\mbox{ s.t. }g(x)=\lambda p\}.\]
In this section, we seek the condition for this $f^g$ matches the utility maximization hypothesis.

First, we introduce three axioms on $g$. The first axiom is as follows: $g$ satisfies the {\bf weak weak axiom} if and only if, for any $x,y\in \Omega$, if $g(x)\cdot y\le g(x)\cdot x$, then $g(y)\cdot x\ge g(y)\cdot y$.\footnote{This axiom is related to the weak weak axiom of revealed preference introduced by Kihlstrom et al. (1976), and thus we call this axiom the ``weak weak axiom''.} It is equivalent to the following assertion: $g(x)\cdot y<g(x)\cdot x$ implies that $g(y)\cdot x>g(y)\cdot y$.

Second, $g$ is said to satisfy the {\bf weak axiom} if and only if, for any $x,y\in \Omega$, if $g(x)\cdot y\le g(x)\cdot x$ and $y\neq x$, then $g(y)\cdot x>g(y)\cdot y$. Clearly, the weak axiom is stronger than the weak weak axiom. Later, we confirm that the weak axiom is needed for assuring $f^g$ to be single-valued.

Next, consider a piecewise $C^1$ closed curve $x:[0,T]\to \Omega$.\footnote{A function $x:[0,T]\to \mathbb{R}^N$ is called a piecewise $C^1$ function if and only if, it is continuous and there exists a finite set $\{t_0,...,t_k\}$ such that $t_0=0, t_k=N$, $t_i<t_{i+1}$, and for $i\in \{1,...,k\}$, there exists a $C^1$ function $g_i: [0,T]\to \mathbb{R}^N$ that coincides with $x$ on $[t_i,t_{i+1}]$. Moreover, a function $x:[0,T]\to \mathbb{R}^N$ is called a {\bf closed curve} if and only if $x(0)=x(T)$.} We call this curve a {\bf Ville curve} if
\[g(x(t))\cdot \dot{x}(t)>0\]
for all $t\in [0,T]$ such that $\dot{x}(t)$ is defined. We say that $g$ satisfies {\bf Ville's axiom} if and only if there is no Ville curve.

Suppose that $f^g=f^u$ for some $C^1$ function $u:\Omega\to \mathbb{R}$ such that $\nabla u(x)\gneq 0$ for all $x\in \Omega$. If $x(t)$ is a Ville curve, then $u(x(t))$ is increasing, and thus $u(x(0))<u(x(T))=u(x(0))$, which is a contradiction. Therefore, $g$ satisfies Ville's axiom. Moreover, if $g(x)\cdot y\le g(x)\cdot x$, then $y\in \Delta(g(x),g(x)\cdot x)$. This implies that $u(x)\ge u(y)$, and thus, for any $a>1$, $u(ax)>u(y)$. Therefore, $ax\notin \Delta(g(y),g(y)\cdot y)$, and thus $g(y)\cdot ax>g(y)\cdot y$. Letting $a\downarrow 1$, we obtain that $g(y)\cdot x\ge g(y)\cdot y$, and thus $g$ satisfies the weak weak axiom. In conclusion, the utility maximization hypothesis implies that both axioms hold. We will show that, this is also true even when $u$ is not differentiable, and the converse is also true.\footnote{See III) of Theorem \ref{THM1}.}

We introduce a condition for a binary relation. A binary relation $\succsim$ on $\Omega$ is said to satisfy the {\bf LNST condition} if, for any $x,y\in \Omega$, if $x\succsim y$ and $U$ is a neighborhood of $x$, then there exists $z\in U$ such that $z\succ y$.

Note that, if $\succsim$ is transitive and locally non-satiated, then it satisfies the LNST condition. We show this result. Suppose that $x\succsim y$ and $U$ is a neighborhood of $x$. Then, by the local non-satiation, there exists $z\in U$ such that $z\succ x$. Because of the definition of $\succ$ and the transitivity of $\succsim$, we have that $z\succsim y$. If $z\not\succ y$, then $y\succsim z$. Because $x\succsim y$, by the transitivity, $x\succsim z$, which is a contradiction. Therefore, $z\succ y$, as desired.

Let $U\subset \Omega$ be open and $u:U\to \mathbb{R}$. We say that $u$ satisfies {\bf property (F)} if and only if it is increasing and locally Lipschitz, and for any $w\in \mathbb{R}$, if $X=u^{-1}(w)$ is not empty, then it is an $n-1$ dimensional $C^1$ manifold, and $g(x)$ is orthogonal to the tangent space $T_x(X)$ for all $x\in X$.\footnote{We say that a real-valued function $u$ is {\bf increasing} if $x\gg y$ implies that $u(x)>u(y)$, and when $u$ is differentiable, then it is said to be {\bf non-degenerate} if $\nabla u(x)\neq 0$ everywhere.} Note that, if $u:U\to \mathbb{R}$ is a $C^1$ increasing non-degenerate function, then $u$ satisfies property (F) if and only if $\nabla u(x)=\lambda(x)g(x)$ for each $x\in U$, where $\lambda$ is some positive real-valued function defined on $U$.

Now, we complete the preparation of stating our first main results.

\begin{thm}\label{THM1}
For a given locally Lipschitz function $g:\Omega\to \mathbb{R}^n_+\setminus\{0\}$, the following four results hold.
\begin{enumerate}[I)]
\item The function $g$ satisfies the weak weak axiom if and only if $f^g=f^{\succsim}$ for a complete binary relation $\succsim$ on $\Omega$ that satisfies the LNST condition.

\item The function $g$ satisfies Ville's axiom if and only if there exists a real-valued function $u:\Omega\to \mathbb{R}$ that satisfies property (F). In this case, for each $v\in \Omega$, there uniquely exists a function $u^g_v:\Omega\to \mathbb{R}$ such that $u^g_v$ satisfies property (F) and $u^g_v(av)=a$ for any $a>0$. If $g$ is $C^k$, then $u^g_v$ is also $C^k$ and $\nabla u^g_v(x)=\lambda(x)g(x)$ for each $x\in \Omega$, where $\lambda$ is some positive real-valued function.

\item The following three statements are equivalent.
\begin{enumerate}[1)]
\item The function $g$ satisfies the weak weak axiom and the Ville's axiom.

\item There exists a quasi-concave function $u:\Omega\to \mathbb{R}$ that satisfies property (F).

\item there exists a weak order $\succsim$ such that $f^g=f^{\succsim}$.
\end{enumerate}
Moreover, in this case, if $u$ satisfies property (F), then $u$ is quasi-concave and $f^g=f^u$. In particular, $u^g_v$ is quasi-concave, and $f^g=f^{u^g_v}$ for any $v\in \Omega$.

\item $g$ satisfies the weak axiom if and only if $f^g$ is single-valued and satisfies the weak axiom of revealed preference. In particular, if $g$ satisfies Ville's axiom, then $g$ satisfies the weak axiom if and only if $u^g_v$ is strictly quasi-concave.
\end{enumerate}
\end{thm}

\noindent
{\bf Remarks on Theorem \ref{THM1}}. We mention the interpretations of axioms. First, suppose that $g(x)\cdot y<g(x)\cdot x$. If we define $v=x-y$, then $g(x)\cdot v>0$. Therefore, the consumer that has the consumption vector $x$ finds that the transaction that moves the consumption vector slightly in the $x-y$ direction is preferable. If $g(y)\cdot x>g(y)\cdot y$, then $g(y)\cdot v>0$. Therefore, when the consumer has the consumption vector $y$, then he/she also finds the transaction that moves the consumption vector slightly in the $x-y$ direction preferable. The weak weak axiom states such a consistency of the consumer's preference.

Next, Ville's axiom is related to Samuelson's (1950) ``three-sided tower'' argument. He considered a consumer in which for three linearly independent vectors $x,y,z\in \Omega$, where $x$ is indifferent to $y$, $y$ is indifferent to $z$, and $z$ is indifferent to $ax$ for some $a<1$. He said that such a consumer is ``easily cheated'' because of the following reasons. First, suppose that the consumer has the consumption vector $x$. If this consumer were offered a transaction that changed the consumption vector from $x$ to $y$, he/she would accept this transaction because he/she feels that $x$ and $y$ are indifferent. Next, if he/she were offered a transaction that changed his/her consumption vector from $y$ to $z$, he/she would accept the transaction for the same reason. Finally, if he/she were offered a transaction that changed his/her consumption vector from $z$ to $ax$, he/she would accept the transaction. As a result of three transactions, the consumer would be forced to tolerate a lower consumption vector than he/she originally had.\footnote{Samuelson apparently did not see any problem with the fact that an indifference curve can be drawn on any two-dimensional space. Therefore, he thought that it was no problem to use the `indifference' relation for $x$ and $y$, $y$ and $z$, and $z$ and $ax$, in his discussions. This view can also be found in Pareto (1906). Like Samuelson, Pareto considered irrational consumers in the mathematical appendix of his book (Pareto, 1909), and his so-called ``open cycle theory'' seems to include a similar idea as Ville curves. Ville's axiom itself is introduced by Ville (1946), and analyzed by Hurwicz and Richter (1979). Note that, although the proof of our Theorem \ref{THM1} also makes extensive use of indifference curves, the construction method is presented rigorously through a differential equation there.}

Actually, for a consumer that has a subjective exchange ratio function $g$ violating Ville's axiom, we can show that there exists the following Ville curve $x:[0,T]\to \Omega$: there exists $t_1,t_2,t_3\in [0,T]$ such that, 1) $x(t)\in \mbox{span}\{x(0),x(t_1)\}$ for $t\in [0,t_1]$, 2) $x(t)\in \mbox{span}\{x(t_1),x(t_2)\}$ for $t\in [t_1,t_2]$, 3) $x(t)\in \mbox{span}\{x(t_2),x(t_3)\}$ for $t\in [t_2,t_3]$, and 4) $x(t_3)=ax(0)$, where $0<a<1$.\footnote{See Step 1 of the proof of Lemma \ref{LEM5}.} If we define $x(0)=x$, $x(t_1)=y$, $x(t_2)=z$, then this consumer considers that $y$ is preferred to $x$, $z$ is preferred to $y$, and $ax$ is preferred to $z$, but $x$ is greater than $ax$. Therefore, we can say that this consumer is also ``easily cheated'' for the same reason as Samuelson's argument.

Let us explain statement II) in Theorem \ref{THM1}. Although property (F) is a little difficult to understand, this property is important for the economic interpretation of this result. To simplify the argument, suppose that $u$ satisfies property (F) and is continuously differentiable. Then. $g(x)$ is positively proportional to $\nabla u(x)$ for each $x\in \Omega$. Hence, the function $\nabla u$ represents the same subjective exchange ratio as $g$. This implies that we can see the subjective value of this consumer as the marginal utility. In fact, this difference in perspective is the transformation of consumer theory brought about by the marginal revolution. In other words, what was thought to be a subjective exchange ratio before the marginal revolution came to be regarded simply as a ratio of marginal utilities after the marginal revolution. And what II) of Theorem \ref{THM1} shows is that Ville's axiom is a necessary and sufficient condition for the justifiability of this view.

By III), the consumer's behavior represented by $f^g$ matches the maximization of the function $u^g_v$, which is a standard hypothesis of consumer behavior in the modern microeconomic theory, if and only if $g$ satisfies the weak weak axiom. Since the quasi-concavity of $u^g_v$ was implicitly assumed at the period of the marginal revolution, the gap between the idea that $\nabla u$ can be viewed as a subjective exchange ratio and the idea of utility maximization may have been overlooked. However, to justify utility maximization, $u^g_v$ must be quasi-concave, and the condition for this is exactly the weak weak axiom.\footnote{See Lemma \ref{LEM3} in the proof section.}

We mention the independence of two axioms. First, as we will show later, if $n=2$, then any $g$ satisfies Ville's axiom. Therefore, if we define $g(x)=x$, then it satisfies Ville's axiom.\footnote{Actually, if $u(x)=\frac{1}{2}[x_1^2+x_2^2]$, then $g(x)=\nabla u(x)$.} However, if $x=(2,1)$ and $y=(1,2)$, then $g(x)\cdot y<g(x)\cdot x$ and $g(y)\cdot x<g(x)\cdot x$, which implies that the weak weak axiom is violated. Second, let $n=3$. Gale (1960) found a function $g$ that satisfies the weak axiom but $f^g$ violates the strong axiom of revealed preference. By Theorem \ref{THM1}, this means that $g$ must violates Ville's axiom. Therefore, these axioms are independent.

\subsection{Dynamic Theory: Stability and Utility Maximization}
Insofar as the consumer determines his/her behavior through utility maximization, his/her behavior can be described statically. That is, the consumption vector that the consumer finally chooses is the one that maximizes utility, which can be obtained by solving the utility maximization problem. In our theory, however, the consumer does not know a utility function, but only a subjective value. Moreover, since a subjective value changes according to the consumption vector currently held, it is possible that the consumer does not know his/her subjective value under a consumption vector that he/she does not currently own. Assume that the consumer can know $g(x)$ if he/she has the consumption vector $x\in \Omega$, but has no further information. In this case, for a given pair $(p,m)$ of the price vector and the income, it is unknown whether the transaction stopping point $f^g(p,m)$ can be found. Therefore, the static theory is not sufficient to describe the consumer's transactions. A dynamic consumer model for finding the transaction stopping point is additionally needed.

In this subsection, we assume that $f^g$ is a single-valued function for simplifying the arguments. We first define the process that describe the consumer's transaction procedure. Fix a pair $(p,m)$ in the domain of $f^g$. We call a continuous function $h:\Omega\to \mathbb{R}^n$ an {\bf improvement direction function} if 1) $g(x)\cdot h(x)>0$ for all $x\in \Delta(p,m)\setminus \{f^g(p,m)\}$, and 2) $p\cdot h(x)\le 0$ when $p\cdot x=m$. For a given improvement policy $h$, consider the following differential equation:
\begin{equation}\label{IMP}
\dot{x}(t)=h(x(t)), x(0)=x_0,
\end{equation}
where $x_0\in \Delta(p,m)$. We call this differential equation the {\bf improvement process} according to $h$.

We define two notions of stability. First, an improvement process is said to be {\bf locally stable} if, there exists an open neighborhood $U$ of $f^g(p,m)$ such that if $x_0\in U$, then there exists a solution $x(t)$ to (\ref{IMP}) defined on $\mathbb{R}_+$, and for any such solutions, $x(t)$ converges to $f^g(p,m)$ as $t\to \infty$. Second, an improvement process is said to be {\bf compact stable} if, for any solution $x(t)$ to (\ref{IMP}) defined on $\mathbb{R}_+$ such that the trajectory of $x(t)$ is included in an compact set $C\subset \Omega$, $x(t)$ converges to $f^g(p,m)$ as $t\to \infty$.

The following is our second main result.

\begin{thm}\label{THM2}
Suppose that $g:\Omega\to \mathbb{R}^n_+\setminus\{0\}$ is locally Lipschitz, and $f^g$ is single-valued. Then, the following two assertions are equivalent.
\begin{enumerate}[1)]
\item $g$ satisfies the weak weak axiom and Ville's axiom.

\item For any $(p,m)$ in the domain of $f^g$ and any improving direction function $h$, the improvement process (\ref{IMP}) according to $h$ is both locally and compact stable.
\end{enumerate}
\end{thm}

\noindent
{\bf Remarks on Theorem \ref{THM2}}. Theorem \ref{THM1} states that, in the static theory, the weak weak axiom and Ville's axiom are necessary and sufficient for justifying the utility maximization hypothesis. In contrast, Theorem \ref{THM2} states that, in the dynamic theory, the same axioms are necessary and sufficient for stability requirements. Because each stability requirement represents the consumer succeeds in finding a termination point of the transaction, we conclude that the condition for the consumer to ensure closing his/her transaction successfully is to match his/her behavior with the utility maximization hypothesis.

We explain why the utility maximization hypothesis is needed for stability. First, suppose that 1) of Theorem \ref{THM2} holds. Define $L(x)=u^g_v(f^g(p,m))-u^g_v(x)$, where $u^g_v$ is defined in II) of Theorem \ref{THM1}. Then, we can easily verify that $L(x)$ is a Lyapunov function of the differential equation (\ref{IMP}). Therefore, $f^g(p,m)$ is stable in both senses. Second, suppose that $g$ violates Ville's axiom. We can show that there exists a Ville curve $x:[0,T]\to \Omega$ that is a solution to some improvement process. Therefore, such an improvement process fails to satisfy compact stability. Third, suppose that $g$ satisfies Ville's axiom but violates the weak weak axiom. Then, $u^g_v$ is not quasi-concave, and thus for some $(p,m)$, $f^g(p,m)$ is not a local maximum point of $u^g_v$. Therefore, for any open neighborhood $U$ of $f^g(p,m)$, there exists $x_0\in U$ such that $u^g_v(x_0)>u^g_v(f^g(p,m))$, which implies that any improvement process fails to satisfy the local stability.

We also explain why we need two concepts of stability. Essentially, what is desired is global stability. That is, it is most desirable that for any $x_0$, there exists a solution defined on $\mathbb{R}_+$ in (1), and any such solution converges to $f^g(p,m)$ as $t\to \infty$. However, the conditions imposed on the improvement direction function $h$ are too weak, and the possibility exists that $x(t)$ can be defined only on a finite time interval. Therefore, global stability cannot be used. Compact stability is requested in place of this global stability. However, there could be a possibility that if $x_0\neq f^g(p,m)$, then there is no solution defined on $\mathbb{R}_+$ such that the orbit of $x(t)$ is contained within the compact set of $\Omega$. For example, suppose that $n=2$, $g(x)=x$, $p=(1,1)$, and $m=2$. In this case, $f^g(p,m)=(1,1)$. If we define $h(x)=x-\frac{(x_1+x_2)^2}{4}(1,1)$, then by the Cauchy--Schwarz inequality, we can show that this $h(x)$ satisfies our requirements. Moreover, $h(x)$ satisfies the compact stability, but the reason is very odd. That is, if $x_0\neq (1,1)$ and $p\cdot x_0=m$, then any solution to (1) reachs some point in $\mathbb{R}^2_+\setminus \mathbb{R}^2_{++}$ in a finite time, and thus $x(t)$ can only be defined on finite time interval. Recall that the concept of stability is interpreted as the condition under which the consumer can find the transaction stopping point. It is clear that a consumer that seeks $f^g(p,m)$ according to the improving process using this $h(x)$ fails to find $f^g(p,m)$, even though this process satisfies the compact stability. As can be seen from this example, compact stability is not sufficient to represent our interpretation of stability. In this view, local stability exactly provides what is missing.

\subsection{Evaluation of Axioms}
Both the weak weak axiom and Ville's axiom have a certain interpretation as an axiom of rationality, as discussed in subsection 3.1. However, in addition to their economic interpretations, axioms have another point of evaluation: how strong they are mathematically. From the definition alone, neither axioms can be ascertained in terms of their strength. Thus, we try to derive equivalent conditions that are relatively easy to evaluate its strength.

First, choose any locally Lipschitz function $g:\Omega\to \mathbb{R}^n_+\setminus\{0\}$. By Rademacher's theorem, $g$ is differentiable almost everywhere. We introduce three conditions on $g$ at $x\in \Omega$.
\begin{description}
\item{(A1)} If $v\cdot g(x)=0$, then $\liminf_{t\downarrow 0}\frac{1}{t}v\cdot (g(x+tv)-g(x))\le 0$.

\item{(A2)} If $v\cdot g(x)=0$ and $v\neq 0$, then $\liminf_{t\downarrow 0}\frac{1}{t}v\cdot (g(x+tv)-g(x))<0$.

\item{(B)} $g$ is differentiable at $x$, and if $i,j,k\in \{1,...,n\}$, then\footnote{In this equation, the variable $x$ is abbreviated.}
\[g_i\left(\frac{\partial g_j}{\partial x_k}-\frac{\partial g_k}{\partial x_j}\right)+g_j\left(\frac{\partial g_k}{\partial x_i}-\frac{\partial g_i}{\partial x_k}\right)+g_k\left(\frac{\partial g_i}{\partial x_j}-\frac{\partial g_j}{\partial x_i}\right)=0.\]
\end{description}

Then, the following result holds.

\begin{thm}\label{THM3}
Suppose that $g:\Omega\to \mathbb{R}^n_+\setminus\{0\}$ is locally Lipschitz. Then, the following three results hold.
\begin{enumerate}[i)]
\item The function $g$ satisfies the weak weak axiom if and only if condition (A1) holds everywhere.

\item If condition (A2) holds everywhere, then $g$ satisfies the weak axiom.

\item $g$ satisfies Ville's axiom if and only if condition (B) holds almost everywhere.
\end{enumerate}
\end{thm}

Suppose that, in addition to the local Lipschitz property, $g_n(x)\neq 0$. Define $\bar{g}(x)=\frac{1}{g_n(x)}g(x)$, and for $i,j\in \{1,...,n-1\}$,
\[a_{ij}(x)=\frac{\partial \bar{g}_i}{\partial x_j}(x)-\frac{\partial \bar{g}_i}{\partial x_n}(x)\bar{g}_j(x).\]
The $(n-1)\times (n-1)$ matrix-valued function $A_g(x)=(a_{ij}(x))_{i,j=1}^{n-1}$ is called the {\bf Antonelli matrix} of $g$. We can show the following proposition.

\begin{prop}\label{PROP1}
Suppose that $g:\Omega\to \mathbb{R}^n_+\setminus \{0\}$ is locally Lipschitz. Then, for each $x\in \Omega$ such that $g$ is differentiable at $x$ and $g_n(x)\neq 0$, the following results hold.
\begin{enumerate}[i)]
\item $g$ satisfies condition (A1) at $x$ if and only if $A_g(x)$ is negative semi-definite.

\item $g$ satisfies condition (A2) at $x$ if and only if $A_g(x)$ is negative definite.

\item $g$ satisfies condition (B) at $x$ if and only if $A_g(x)$ is symmetric.
\end{enumerate}
\end{prop}

Using these results, we can evaluate the strength of Ville's axiom. Let $\mathscr{G}$ be the set of all locally Lipschitz function $g:\Omega\to \mathbb{R}^n_+\setminus \{0\}$ such that $\sum_{i=1}^ng_i(x)=1$ for all $x\in \Omega$.

\begin{cor}\label{COR1}
If $n=2$, then every $g\in \mathscr{G}$ satisfies Ville's axiom. If $n\ge 3$, the set of all $g\in \mathscr{G}$ that satisfies Ville's axiom is nowhere dense with respect to any topology of $\mathscr{G}$ such that the function $(1-t)g_1+tg_2$ is continuous in $t$.
\end{cor}

Hence, we conclude that Ville's axiom is, at least mathematically, strong whenever $n\ge 3$. In contrast, we could not derive a similar result on the weak weak axiom. The reason is demonstrated in Proposition \ref{PROP1}. Ville's axiom corresponds to the symmetry of the Antonelli matrix. If $n=2$, then the Antonelli matrix is $1\times 1$ matrix, which is automatically symmetric. If $n\ge 3$, however, the Antonelli matrix is $(n-1)\times (n-1)$ matrix, and in the space of such matrices, the set of all symmetric matrices is a null set. In contrast, the weak weak axiom corresponds to the negative semi-definiteness of the Antonelli matrix, and even when $n\ge 3$, the set of negative semi-definite matrices is not null. Therefore, the set of $g$ that satisfies the weak weak axiom may have a nonempty interior.

Of course, the fact that an axiom is strong does not immediately mean that it should not be adopted. However, to recognize the fact that Ville's axiom is mathematically quite strong is important. Whenever we model a human through utility maximization, we are placing this strong assumption.

\subsection{Embedding Subjective Exchange Ratios into the Space of Preferences}
In this paper, we discussed the theory of preference relations and the theory of subjective values separately. In fact, however, for any locally Lipschitz function $g$, we can construct a naturally corresponding preference relation $\succsim^g$. The construction is one-to-one and continuous, and, moreover, unique in a certain sense. Although this result will be verified and heavily used as a lemma in the proofs of our main results, we here describe this result itself as a theorem.

We define an additional condition on binary relations. Suppose that $\succsim$ is a binary relation on $\Omega$. Then, it is said to be {\bf p-transitive} if for any $x,y,z\in \Omega$, $\dim(\mbox{span}\{x,y,z\})\le 2$, $x\succsim y$, and $y\succsim z$ imply $x\succsim z$. Any transitive preference is p-transitive, and the converse is true only when $n=2$.

Now, suppose that $g:\Omega\to \mathbb{R}^n_+\setminus \{0\}$ is locally Lipschitz, and choose $x,v\in \Omega$ and consider the following differential equation:
\[\dot{y}(t)=(g(y(t))\cdot x)v-(g(y(t))\cdot v)x,\ y(0)=x.\]
Let $y(t;x,v)$ be the solution function of the above equation. We will prove in the proof section that there exists $t(x,v)\ge 0$ such that $y(t(x,v);x,v)$ is proportional to $v$. Define
\[x\succsim^gv\Leftrightarrow y(t(x,v);x,v)\ge v.\]
Then, the following theorem holds.

\begin{thm}\label{THM4}
For any locally Lipschitz function $g:\Omega\to \mathbb{R}^n_+\setminus \{0\}$, $\succsim^g$ is a complete, p-transitive, continuous, and monotone binary relation on $\Omega$, and the following results hold.
\begin{enumerate}[I)]
\item If $f^g=f^{\succsim}$ for some complete, p-transitive, and continuous binary relation, then $\succsim=\succsim^g$.

\item The following three conditions are equivalent.
\begin{enumerate}[i)]
\item The function $g$ satisfies the weak weak axiom.

\item $\succsim^g$ is weakly convex.

\item $f^g=f^{\succsim^g}$.
\end{enumerate}

\item The following three conditions are equivalent.
\begin{enumerate}[i)]
\item The function $g$ satisfies Ville's axiom.

\item $\succsim^g$ is transitive.

\item $u^g_v$ represents $\succsim^g$, where $u^g_v$ is defined in II) of Theorem \ref{THM1}.
\end{enumerate}

\item If $h:\Omega\to \mathbb{R}^n_+\setminus\{0\}$ is another locally Lipschitz function, then $\succsim^g=\succsim^h$ if and only if $h(x)$ is proportional to $g(x)$ everywhere.

\item The mapping $g\mapsto \succsim^g$ is continuous on the pair of the topology of compact convergence and the closed convergence topology.
\end{enumerate}
\end{thm}

\noindent
{\bf Remarks on Theorem \ref{THM4}}. As I)-III) implies, $\succsim^g$ naturally corresponds to $g$. In particular, by II), the weak weak axiom of $g$ corresponds to the weak convexity of $\succsim^g$, and by III), Ville's axiom of $g$ corresponds to the transitivity of $\succsim^g$, respectively. Thus, some rationality conditions of $g$ are equivalent to that of $\succsim^g$. In this sense, $\succsim^g$ is a natural preference representation of $g$.

We believe that $g_i(x)$ represents the subjective value of the $i$-th commodity. However, this value has no unit. Therefore, not the absolute level of the value $g_i(x)$, but its ratio is most important, because it represents the subjective exchange ratio. In this sense, $h(x)$ in IV) expresses exactly the same subjective exchange ratio as $g(x)$. Thus, economically speaking, it is natural that $g$ and $h$ correspond to the same preference. On the other hand, if $h(x)$ and $g(x)$ express different subjective exchange ratios, it is natural that they correspond to different preferences. IV) of Theorem \ref{THM4} implies that these natural requirements hold. Furthermore, V) expresses that if the difference between two subjective exchange ratios is small, then corresponding preferences also have only a small difference. In this sense, we can say that $\succsim^g$ is the preference expression that naturally corresponds to $g$.

\section{Comparison with Related Literature}
The research topic of this paper was already known as the integrability problem at the end of the 19th century. Antonelli (1886) was the first monograph to tackle this problem, and in this paper, there is a reference to condition (B). Pareto (1906) also addressed this problem, but Volterra was critical of Pareto's discussion of integrability in his review (Volterra, 1906). Hence, Pareto tried to develop a consumer theory without condition (B) in the French version of his book (Pareto, 1909). Pareto argued that condition (B) is related to the order of consumption. This idea is later implicitly criticized by Samuelson (1950). Suda (2007) surveyed these arguments in detail.

The weak axiom and the weak weak axiom are paraphrases of the corresponding conditions in the demand function, respectively. The former was first treated by Samuelson (1938), and it was shown that the Slutsky matrix is negative semi-definite under this condition. The latter was first presented by Kihlstrom et al. (1976), and it was shown that this condition is equivalent to the negative semi-definiteness of the Slutsky matrix under the assumption that the demand function is differentiable. Samuelson (1950) showed that, if both $g$ and $f^g$ are differentiable, then the Antonelli matrix is the inverse of a submatrix of the Slutsky matrix. Thus, condition (A2) is equivalent to the negative semi-definiteness of the Slutsky matrix, and condition (B) is equivalent to the symmetry of the Slutsky matrix. Condition (A1) does not have a corresponding property of the Slutsky matrix, since it is worthwhile only when $f^g$ is not differentiable.\footnote{Indeed, if $g$ and $f^g$ are differentiable and $g$ satisfies condition (B), then $A_g(x)$ is symmetric and invertible, and thus it is negative semi-definite if and only if it is negative definite. Hence, condition (A1) is equivalent to condition (A2).}

Ville's axiom was introduced by Ville (1946). Hurwicz and Richter (1979) showed that if $g$ is continuously differentiable, then this axiom is equivalent to condition (B). However, Ville's axiom is more essential for our results because we can show Theorem \ref{THM1} with no use of condition (B). In the proof section, we show that if $g$ does not satisfy Ville's axiom, then there exists a Ville curve such that there exists an increasing finite sequence $t_0,t_1,t_2,t_3$ such that $t_0=0$, $x(t_3)=ax(t_0)$ for some $a<1$, and $x(t)\in \mbox{span}\{x(t_i),x(t_{i+1})\}$ when $t\in [t_i,t_{i+1}]$. Samuelson (1950) described this as what happens to a consumer who violates condition (B), and he called such an individual ``easily cheated''. Our interpretation of Ville's axiom follows this, and a consumer who satisfies Ville's axiom is seen as ``hardly cheated'' in this sense.

Hurwicz (1971) explained that integrability theory is the theory that finds conditions for demand correspondence to be associated with utility maximization. This explanation, however, is problematic in two ways. First, if this explanation is correct, then revealed preference theory is a part of integrability theory. However, most researchers would disagree with this view. Second, although integrability theory is a theory that has existed since the 19th century, it is unlikely that the concept of ``demand correspondence'' existed when Antonelli wrote his monograph. Therefore, integrability theory in Hurwicz's explanation is at least separated from the classical integrability theory.

Hurwicz and Uzawa (1971) is considered to be the leading modern treatise on integrability theory. In their introduction, they classify integrability theory into direct and indirect approaches. Most of the traditional work on integrability theory is classified as the indirect approach in their classification. On the other hand, Hurwicz--Uzawa theory is a direct approach, which characterizes utility maximization in terms of Slutsky matrix. The fact that it deals with Slutsky matrix means that the most important object of analysis in this theory is the demand function, and the subjective value function does not appear there. In this sense, the research objective of their theory is different from that of integrability theory in the 19th century.

Debreu (1972), on the other hand, considered a method to obtain the corresponding twice continuously differentiable utility function $u$ starting from a continuously differentiable subjective value function $g$ that satisfies condition (B), and used it to discuss the conditions for the demand function to become differentiable. The reason why he wanted such a result is as follows. When he discussed the discreteness of the set of equilibrium prices in Debreu (1970), he assumed the differentiability of the demand function. However, Katzner (1968) found an example where the demand function is not differentiable even though the corresponding utility function is smooth, and thus, Debreu needed a condition for a demand function to become differentiable. Debreu's condition for the demand function to be differentiable can be verified to be correct by a method that has nothing to do with the integrability problem. On the other hand, the part about finding $u$ from $g$ is known to contain some error. In his corrigendum (Debreu, 1976), he explained that, contrary to his original claim, there are cases where $g$ is continuously differentiable and it has no corresponding twice continuously differentiable $u$, but he did not explain which part of the proof is problematic. We too cannot prove that $u^g_v$ in Theorem \ref{THM1} is $C^{k+1}$, even if $g$ is $C^k$.

The construction of $u^g_v$ in Theorem \ref{THM1} has already been shown by Hosoya (2013), provided that $g$ is continuously differentiable and $g(x)\in \mathbb{R}^n_{++}$. However, that proof is considerably different from the proof in this paper. In the first place, neither the weak weak axiom nor Ville's axiom are treated in Hosoya (2013), and what are treated instead are conditions (A1) and (B). The proofs are also radically different. Hosoya (2013) used condition (A1) to prove the quasi-concavity of $u^g_v$, but what was discussed there was essentially the characterization theorem by differentiation of the quasi-concavity presented in Otani (1983). In contrast, in this paper, we mainly use the fact that the approximate polygonal line of the indifference curve constructed by Euler's method is convex toward the origin under the weak weak axiom. These are two completely different arguments.

Moreover, because Ville's axiom treated in this paper is too strong, it is not equivalent to condition (B) if the range of the vector field $g$ is not necessarily contained in $\mathbb{R}^n_+$. A counterexample to this is given in Debreu (1972). In this connection, we mention the known result related to condition (B). If $g$ is $C^1$, then condition (B) is known as a necessary and sufficient condition for the local existence of a real-valued function $u$ whose gradient vector field is positively proportional to $g$ at each point, and this result is called Frobenius' theorem. This theorem was already known in the 19th century, and was essentially important for results in Debreu (1972) and Hosoya (2013). Although Frobenius' theorem only treats the continuously differentiable vector fields, when $g$ is locally Lipschitz, then a similar result is provided in Hosoya (2021). We use this extension of Frobenius' theorem in the proof of Theorem \ref{THM3}. However, this theorem is not used in the proof of Theorems \ref{THM1} and \ref{THM2}. In this sense, the proof of our main results are completely different from that in Debreu (1972) or Hosoya (2013).

Although we imposed the local Lipschitz property for $g$, readers may think that only the continuity of $g$ should be imposed. However, we know that if $g$ is not locally Lipschitz, then there may be multiple candidates for the indifference curve passing through $x$. As a result, I) of Theorem \ref{THM4} can never be derived, nor can Theorem \ref{THM1} be proved in the first place. An actual case in which I) of Theorem \ref{THM4} cannot be derived was provided in Mas-Colell (1977).

In this connection, we mention the uniqueness theorem provided in Hosoya (2020). This paper proved that under the assumption that there exists a weak order $\succsim$ such that $f=f^{\succsim}$, if $f$ satisfies the {\bf income-Lipschitzian property}, then there uniquely exists an upper semi-continuous weak order $\succsim^*$ such that $f=f^{\succsim^*}$. On the other hand, I) of Theorem \ref{THM4} do not assume anything on $f^g$ and derive the uniqueness of the continuous weak order $\succsim^g$ that satisfies $f^g=f^{\succsim^g}$ from the local Lipschitz property of $g$. These results may possibly be independent. Unfortunately, we do not know if there exists a function $g$ such that $f^g$ does not satisfy the income-Lipschitzian requirement but $g$ is locally Lipschitz. Conversely, there is an example of $f^u$ such that there is no single-valued continuous function $g$ such that $f^g=f^u$, but it is income-Lipschitzian: see example of Hurwicz and Uzawa (1971). Therefore, we at least know that for this $f^u$, the result of Hosoya (2020) is applicable, while I) of Theorem \ref{THM4} is not applicable.

In connection with Theorem \ref{THM2}, it is unknown whether the same result can be obtained by assuming that $p\cdot x_0=m$, $p\cdot h(x)=0$ for all $x\in \Delta(p,m)$ such that $p\cdot x=m$, and the solution to (\ref{IMP}) always satisfies $p\cdot x(t)=m$. In this case, a similar result can still be derived if $g$ satisfies Ville's axiom. In other words, if Ville's axiom is satisfied, then both stabilities hold if and only if $g$ satisfies the weak weak axiom. However, it is unknown whether there exists a Ville curve that always satisfies $p\cdot x(t)=m$ if $g$ violates Ville's axiom. This is one of the open problems related to this paper. Perhaps the theory of de Rham cohomology on manifolds may be useful. See chapter 4 of Guillemin and Pollack (1976) for more detailed arguments.

For Theorem \ref{THM3}, it is known that if there exists a twice continuously differentiable function $u$ such that $g=\nabla u$, then condition (B) holds. Moreover, in this case, condition (A1) is equivalent to $v^TD^2u(x)v\le 0$ whenever $\nabla u(x)\cdot v=0$, and condition (A2) is equivalent to $v^TD^2v<0$ whenever $\nabla u(x)\cdot v=0$ and $v\neq 0$. The former is called the bordered Hessian condition, and the latter is called the strict bordered Hessian condition. When $u$ is twice continuously differentiable and $\nabla u(x)\neq 0$ for all $x$, the quasi-concavity of $u$ is equivalent to the bordered Hessian condition (Otani, 1983), and the continuous differentiability of $f^u$ is equivalent to the strict bordered Hessian condition (Debreu, 1972). Recently, Hosoya (2022) showed that when $u$ and $g$ are continuously differentiable and $\nabla u$ is positively proportional to $g$, condition (A1) is also equivalent to the quasi-concavity of $u$, and $u$ is strictly quasi-concave if $g$ satisfies condition (A2). However, no such result was known when $g$ is not differentiable, and thus Theorem \ref{THM3} is a new result.

The continuity of the mapping $g\mapsto \succsim^g$ has been proved in Hosoya (2015) under the assumption that $g$ is continuously differentiable. However, in this paper, $g$ is not necessarily continuously differentiable, and thus this result is a new result.

Finally, we discuss the relationship between our axioms and axioms of revealed preference. Richter (1966) showed that $f^g=f^{\succsim}$ for some weak order $\succsim$ if and only if $f^g$ satisfies the congruence axiom of revealed preference. By Theorem \ref{THM1}, we see that $f^g$ satisfies the congruence axiom of revealed preference if and only if $g$ satisfies both the weak weak axiom and Ville's axiom. Also, by IV) of Theorem \ref{THM1} and an easy calculation, we can show that $g$ satisfies the weak axiom if and only if $f^g$ is single-valued and satisfies the weak axiom of revealed preference. On the other hand, it is not easy to characterize the case where $g$ satisfies the weak weak axiom but not the weak axiom. In this case, we can easily show that there exists an indifference curve of $\succsim^g$ in Theorem \ref{THM4} that contains a line segment. However, since $\succsim^g$ is not necessarily transitive, we cannot show that $f^g=f^{\succsim^g}$ is not single-valued. If $f^g$ is single-valued, then $f^g$ satisfies the weak weak axiom of revealed preference presented in Kihlstrom et al. (1976). If $f^g$ is not single-valued, then it satisfies the weak axiom of revealed preference as a multi-valued choice function.

\section{Conclusion}
In this paper, we develop a consumer theory based on subjective exchange ratios and derive axioms for its consistency with the utility maximization hypothesis. We also confirmed that the conditions for a consumer who can only see subjective values under which he/she can find a transaction stopping point are the same axioms to be consistent with the utility maximization hypothesis.

In addition to these results, we derived equivalence conditions of axioms for considering their mathematical strength. Based on these equivalence condition, we confirmed that Ville's axiom is very strong when the number of commodities is more than two. We also discussed how to express our theory of subjective exchange ratios in terms of binary relations.

These results allow us to adequately explain the additional conditions implicitly imposed on consumer behavior by the utility maximization hypothesis. Although some open problems remain, we believe that this paper has allowed us to better understand the place of the utility maximization hypothesis in the theory of consumer behavior.

\section{Proofs}
\subsection{Mathematical Preliminary}
We repeatedly use the property of the solution function of differential equations, Lipschitz analysis, and differential topology in the proof section. In this subsection, we introduce several important facts without proof.

First, we explain some knowledge of ordinary differential equations (ODEs). Consider the following ODE:
\begin{equation}\label{eq6.00}
\dot{x}(t)=h(t,x(t)),\ x(t_0)=x^*,
\end{equation}
where $h:U\to \mathbb{R}^N$ and $U\subset \mathbb{R}\times\mathbb{R}^N$ is open. We call a subset $I$ of $\mathbb{R}$ an {\bf interval} if it is a convex set containing at least two points. We say that a function $x:I\to \mathbb{R}^N$ is a {\bf solution} to (\ref{eq6.00}) if and only if 1) $I$ is an interval containing $t_0$, 2) $x(t_0)=x^*$, 3) $x$ is continuously differentiable, 4) the graph of $x$ is included in $U$, and 5) $\dot{x}(t)=h(t,x(t))$ for every $t\in I$. Let $x:I\to \mathbb{R}^N$ and $y:J\to \mathbb{R}^N$ be two solutions. Then, we say that $x$ is an {\bf extension} of $y$ if $J\subset I$ and $y(t)=x(t)$ for all $t\in J$. A solution $x:I\to \mathbb{R}^N$ is called a {\bf nonextendable solution} if there is no extension except $x$ itself.

\begin{fact}\label{FACT1}
Suppose that $h$ is locally Lipschitz. Then, for every interval $I$ including $t_0$, there exists at most one solution to (\ref{eq6.00}) defined on $I$. In particular, there exists a unique nonextendable solution $x:I\to \mathbb{R}^N$ to (\ref{eq6.00}), where $I$ is open and $x(t)$ is continuously differentiable. Moreover, for every compact set $C\subset U$, there exist $t_1,t_2\in I$ such that if $t\in I$ and either $t<t_1$ or $t_2<t$, then $(t,x(t))\notin C$.\footnote{For a proof, see Theorems 1.1 and 3.1 in chapter 2 of Hartman (1997).}
\end{fact}

Next, consider the following parametrized ODE:
\begin{equation}\label{eq6.01}
\dot{x}(t)=k(t,x(t),y),\ x(t_0)=z,
\end{equation}
where $k:U\to \mathbb{R}^N$ and $U\subset \mathbb{R}\times \mathbb{R}^N\times \mathbb{R}^M$ is open. We assume that $k$ is locally Lipschitz. Fix $(y,z)$ such that $(t_0,z,y)\in U$. Then, (\ref{eq6.01}) can be seen as (\ref{eq6.00}), where $h(t,x)=k(t,x,y)$ and $x^*=z$. Hence, we can define a nonextendable solution $x^{y,z}:I\to \mathbb{R}^N$ according to Fact \ref{FACT1}. We write $x(t;y,z)=x^{y,z}(t)$, and call this function $x:(t,y,z)\mapsto x(t;y,z)$ the {\bf solution function} of (\ref{eq6.01}).

\begin{fact}\label{FACT2}
Under the assumption that $k$ is locally Lipschitz, the domain of the solution function of (\ref{eq6.01}) is open, and the solution function is locally Lipschitz.\footnote{For a proof, See Hosoya (2024).}
\end{fact}

\begin{fact}\label{FACT3}
If $k$ is $C^k$, then the solution function of (\ref{eq6.01}) is $C^k$.\footnote{For a proof, see Section 5.3 of Hartman (1997).}
\end{fact}

Next, we extend the smoothness of functions. Let $f:X\to \mathbb{R}^M$, where $X\subset \mathbb{R}^N$. We say that $f$ is $C^k$ if and only if, for each $x\in X$, there exist an open neighborhood $U$ of $x$ in $\mathbb{R}^N$ and a $C^k$ function $F:U\to \mathbb{R}^M$ such that $F(x)=f(x)$ for all $x\in U\cap X$. If $f:X\to Y$ is a bijection and both $f$ and $f^{-1}$ are $C^k$, then we call $f$ a $C^k$ {\bf diffeomorphism}. A set $X\subset \mathbb{R}^N$ and $Y\subset \mathbb{R}^M$ are said to be $C^k$ {\bf diffeomorphic} if there exists a $C^k$ diffeomorphism $f:X\to Y$.

A nonempty set $X\subset \mathbb{R}^N$ is called an {\bf $L$ dimensional $C^k$ manifold} if, for every $x\in X$, there exists a neighborhood $U$ of $x$ in $X$ that is $C^k$ diffeomorphic to an open set $V\subset \mathbb{R}^L$. A corresponding diffeomorphism $\Phi:V\to U$ is called a {\bf local parametrization} around $x$. If $\Phi:V\to U$ is a $C^k$ local parametrization around $x$ and $\Phi(z)=x$, then $T_x(X)\equiv \{D\phi(z)v|v\in \mathbb{R}^L\}$ is called the {\bf tangent space} of $X$ at $x$. The following result holds.

\begin{fact}\label{FACT4}
If $X$ is an $L$ dimensional $C^k$ manifold, then any tangent space of $X$ is an $L$ dimensional linear space of $\mathbb{R}^N$ that is independently determined from the choice of the local parametrization. Let $C^1(x,X)$ be the set of all $C^1$ functions $c:I\to X$, where $I$ is an open interval including $0$ and $c(0)=x$. Then, $T_x(X)=\{\dot{c}(0)|c(t)\in C^1(x,X)\}$.
\end{fact}

The next fact is known as the {\bf $\varepsilon$-neighborhood theorem}.\footnote{For a proof, see Section 2.3 of Guillemin and Pollack (1974).}

\begin{fact}\label{FACT5}
Suppose that $X\subset \mathbb{R}^N$ is an $L$ dimensional compact $C^1$ manifold. Then, there exist $\varepsilon>0$ and a $C^1$ function $\pi:U\to X$ such that $\pi(x)=x$ for all $x\in X$, where $U=\{y\in \mathbb{R}^N|\exists x\in X\mbox{ s.t. }\|x-y\|<\varepsilon\}$.
\end{fact}

We also mention an extension of a classical result called Frobenius' Theorem. Suppose that $N\ge 2$, and let a function $g:U\to \mathbb{R}^N\setminus \{0\}$ be given, where $U\subset \mathbb{R}^N$ is an open set. Consider the following {\bf total differential equation}:
\begin{equation}\label{TDE}
\nabla u(x)=\lambda(x)g(x).
\end{equation}
A pair of real-valued functions $(u,\lambda)$ defined on an open set $V\subset U$ is called a {\bf normal solution to (\ref{TDE}) on $V$} if and only if, 1) $u:V\to \mathbb{R}$ is a locally Lipschitz function such that the level set $u^{-1}(a)$ is either the empty set or an $N-1$ dimensional $C^1$ manifold, 2) $\lambda:V\to \mathbb{R}$ is positive, and 3) (\ref{TDE}) holds for almost every $x\in V$. If $V$ is an open neighborhood of $x^*$, then $(u,\lambda)$ is also called a normal solution to (\ref{TDE}) around $x^*$.

\begin{fact}\label{FACT6}
Suppose that $U\subset \mathbb{R}^N$ is open and $g:U\to \mathbb{R}^N\setminus \{0\}$ is locally Lipschitz. Then, $g$ satisfies condition (B) almost everywhere if and only if, for any $x^*\in U$, there exists a normal solution $(u,\lambda)$ around $x^*$. In this case, for a $C^1$ manifold $X=u^{-1}(a)$ and each $x\in X$, $g(x)$ is normal to the tangent space $T_x(X)$. Furthermore, if $g$ is $C^k$, there exists a normal solution $(u,\lambda)$ around $x^*$ such that $u$ is $C^k$.\footnote{For a proof, See Hosoya (2021).}
\end{fact}

Now, recall the definition of a locally Lipschitz function. Let $f:U\to \mathbb{R}^N$ be some function, where $U\subset \mathbb{R}^M$ is open. This function is said to be {\bf locally Lipschitz} if, for every compact set $C\subset U$, there exists $L>0$ such that for every $x,y\in C$,
\[\|f(x)-f(y)\|\le L\|x-y\|.\]
Then, the following fact holds.\footnote{For a proof, see Hosoya (2024).}

\begin{fact}\label{FACT7}
Let $f:U\to \mathbb{R}^N$, where $U\subset \mathbb{R}^M$ is open. Then, $f$ is locally Lipschitz if and only if, for every $x\in U$, there exists $r>0$ and $L>0$ such that if $y,z\in U,\ \|y-x\|\le r,$ and $\|z-x\|\le r$, then
\[\|f(y)-f(z)\|\le L\|y-z\|.\]
\end{fact}

The next fact is known as Gronwall's inequality.\footnote{For a proof, see Lemma 1 of Hosoya (2024).}

\begin{fact}\label{FACT8}
Suppose that $g:[t_0,t_1]\to \mathbb{R}$ is continuous, and
\[g(t)\le \int_{t_0}^tAg(s)ds+B(t),\]
for almost every $t\in [t_0,t_1]$, where $A>0$ and $B(t)$ is an integrable function on $[t_0,t_1]$. Then, for every $t\in [t_0,t_1]$,
\[g(t)\le B(t)+A\int_{t_0}^te^{A(t-s)}B(s)ds.\]
In particular, if $B(t)=C(t-t_0)$ for some constant $C$, then for every $t\in [t_0,t_1]$,
\[g(t)\le \frac{C}{A}(e^{A(t-t_0)}-1).\]
\end{fact}

The next fact is proved in Lemma 3.1 of Hosoya (2021).

\begin{fact}\label{FACT9}
Suppose that $n\ge 2$, $U\subset \mathbb{R}^{n+1}$ is an open and convex set, $W\subset U$, and $U\setminus W$ is a null set. Moreover, suppose that $f:U\to \mathbb{R}^n$ is locally Lipschitz. Let $(x^*,u^*)\in U$, and consider the following differential equation
\begin{equation}\label{SH}
\dot{c}(t)=f((1-t)x^*+tx,c(t))\cdot (x-x^*),\ c(0)=w.
\end{equation}
Let $c(t;x,w)$ be the solution function of (\ref{SH}). Suppose that $x\in \mathbb{R}^n$, $x_{i^*}\neq x_{i^*}^*$, and the domain of the solution function $c(t;x,w)$ of (\ref{SH}) includes $[0,1]\times \bar{P}_1\times \bar{P}_2$, where $P_1$ is a bounded open neighborhood of $x$ and $P_2$ is a bounded open neighborhood of $u^*$, and $\bar{P}_j$ denotes the closure of $P_j$. For every $(t,\tilde{y},w)\subset \mathbb{R}^{n+1}$ such that $t\in [0,1]$, $y\in P_1$ for
\[y_i=\begin{cases}
\tilde{y}_i & \mbox{if }i<i^*,\\
x_i & \mbox{if }i=i^*,\\
\tilde{y}_{i-1} & \mbox{if }i>i^*,
\end{cases}\]
 and $w\in P_2$, define
\begin{equation}\label{XI}
\xi(t,\tilde{y},w)=((1-t)x^*+ty,c(t;y,w)).
\end{equation}
Then, $\xi^{-1}(U\setminus W)$ is also a null set.
\end{fact}

\subsection{Proof of Fundamental Lemmas}
We first introduce several lemmas.

First, we need several definitions of symbols. First, suppose that $x,v\in \Omega$. Define\footnote{The symbols defined in this subsection depend on $x$ and $v$. If needed, we write these as functions (e.g. $P(x,v)$ instead of $P$) for clarifying the meaning.}
\[a_1=\frac{1}{\|x\|},\]
\[a_2=\begin{cases}
\frac{1}{\|v-(v\cdot a_1)a_1\|}(v-(v\cdot a_1)a_1) & \mbox{if }x\mbox{ is not proportional to}v,\\
0 & \mbox{otherwise}.
\end{cases}\]
Second, let $V=\mbox{span}\{x,v\}$, and define functions $P:\Omega\to V$ and $R:V\to V$ such that
\[Py=(y\cdot a_1)a_1+(y\cdot a_2)a_2,\]
\[Rw=(w\cdot a_1)a_2-(w\cdot a_2)a_1.\]
Third, define
\begin{align*}
v_1=&~\arg\min\{w\cdot a_1|w\in P\mathbb{R}^n_+,\|w\|=1,w\cdot a_2\ge 0\},\\
v_2=&~\arg\min\{w\cdot a_1|w\in P\mathbb{R}^n_+,\|w\|=1,w\cdot a_2\le 0\}.
\end{align*}
Fourth, define
\begin{align*}
y_1=&~\{s_1v|s_1\in\mathbb{R}\}\cap \{x+s_2Rv_1|s_2\in\mathbb{R}\},\\
y_2=&~\{s_3v|s_3\in\mathbb{R}\}\cap \{x+s_4Rv_2|s_4\in\mathbb{R}\}.
\end{align*}
Fifth, define
\[\Delta=\{w\in V|w\cdot Rv\le 0,\ w\cdot v_1\ge x\cdot v_1,\ w\cdot v_2\le x\cdot v_2\}.\]
Last, define
\[C=\|x\|\|v-(v\cdot a_1)a_1\|.\]
Then, we can show the following result.

\begin{lem}\label{LEM1}
All the above symbols are well-defined. Moreover, the following results hold.
\begin{enumerate}[i)]
\item If $x$ is not proportional to $v$, then $\{a_1,a_2\}$ is the orthonormal basis of $V$ derived from $x,v$ by the Gram-Schmidt method. Otherwise, then $a_2=0$. In both cases, $P$ is the orthogonal projection from $\mathbb{R}^n$ onto $V$.\footnote{As a result, $y\cdot w=Py\cdot w$ for any $y\in \mathbb{R}^n$ and $w\in V$.}

\item If $x$ is not proportional to $v$, then $R$ is the unique orthogonal transformation on $V$ such that $Ra_1=a_2$ and $Ra_2=-a_1$. Moreover, if $T$ is an orthogonal transformation on $V$ such that $w\cdot Tw=0$ for any $w\in V$, then we must have either $T=R$ or $T=-R=R^{-1}=R^3$.\footnote{In particular, if $z\in V\cap \Omega$ and $[x,z]\cap \{cv|c\in \mathbb{R}\}=\emptyset$, then $R(x,v)=R(z,v)$. Note that $[x,z]$ represents $\{(1-t)x+tz|t\in [0,1]\}$.} If $x$ is proportional to $v$, then $Rw\equiv 0$.

\item If $x$ is not proportional to $v$, then both $v_1$ and $v_2$ are continuous and single-valued at $(x,v)$. If $x$ is proportional to $v$, then $v_1=v_2=a_1$. In both cases, $P\mathbb{R}^n_+=\{c_1v_1+c_2v_2|c_1,c_2\ge 0\}$.

\item Both $y_1$ and $y_2$ are continuous and single-valued. Moreover, for any $(x,v)\in \Omega^2$, $y_1(x,v),\ y_2(x,v)\in \Omega$ and each $y_i(x,v)$ is proportional to $v$. In particular, if $x$ is proportional to $v$, then $y_1=y_2=x$.

\item $\Delta=(x+RP\mathbb{R}^n_+)\cap \{w\in V|w\cdot Rv\le 0\}=\mbox{co}\{x,y_1,y_2\}$.\footnote{By this result, we have that $\Delta$ is a compact subset of $\Omega$.}

\item $(y\cdot x)v-(y\cdot v)x=CRPy$ for any $y\in \mathbb{R}^n$.
\end{enumerate}
\end{lem}

\begin{proof}
For the case in which $x$ is not proportional to $v$, the same claims as this lemma are proved in the proof of Theorem \ref{THM1} of Hosoya (2013). Hence, we only treat the case in which $x$ is proportional to $v$. First, $a_2=0$ and $Rw\equiv 0$ by definition. Second, $P\mathbb{R}^n_+=\{ax|a\ge 0\}$, and thus $v_1=v_2=a_1$. Third, because $\{x+sRv_i|s\in \mathbb{R}\}=\{x\}$, we have that $y_1=y_2=x$. The continuity of $y_i(x,v)$ on $\Omega^2$ is shown in Hosoya (2013). Fourth, by definition, $\Delta$ is the set of $z\in \{ax|a\ge 0\}$ such that $\|z\|\ge \|x\|$ and $\|z\|\le \|x\|$, which means that $\Delta=\{x\}$. Moreover, $x+RP\mathbb{R}^n_+=\{x\}$. Last, $C=0$ and $(y\cdot x)v-(y\cdot v)x=0$ for all $y$. Therefore, all of our claims hold. This completes the proof. $\blacksquare$
\end{proof}

Next, consider the following differential equation:
\begin{equation}\label{IC}
\dot{y}(t)=(g(y(t))\cdot x)v-(g(y(t))\cdot v)x,\ y(0)=x.
\end{equation}
By vi) of Lemma \ref{LEM1}, this equation can be rewritten as follows.
\[\dot{y}(t)=CRPg(y(t)),\ y(0)=x.\]
Let $y(t;x,v)$ be the solution function of (\ref{IC}). Note that, if $w\cdot x=w\cdot v=0$, then $w\cdot \dot{y}(t;x,v)=0$, and thus $w\cdot y(t;x,v)=0$ for all $t$. Therefore, we have that $y(t;x,v)\in V$ for all $t$.

Choose any $x,v\in \Omega$ such that $x$ is not proportional to $v$, and define
\[w^*=(v\cdot x)v-(v\cdot v)x.\]
Since $V$ is two-dimensional, we have that for $y\in V$, $w^*\cdot y=0$ if and only if $y$ is proportional to $v$. Because $v\in V$, we have that $Pv=v$, and because of vi) of Lemma \ref{LEM1}, $w^*=CRv$. Therefore,
\[w^*\cdot \dot{y}(t;x,v)=C^2(v\cdot Pg(y(t;x,v)))=C^2(v\cdot g(y(t;x,v)))>0,\]
for each $t$, and thus $t\mapsto w^*\cdot y(t;x,v)$ is an increasing function. Moreover,
\[w^*\cdot y(0;x,v)=w^*\cdot x=(v\cdot x)^2-\|v\|^2\|x\|^2<0\]
by the Cauchy--Schwarz inequality. Because $g(y(t;x,v))\in \mathbb{R}^n_+$, by iii) of Lemma \ref{LEM1}, $Pg(y(t;x,v))=c_1v_1+c_2v_2$ for some $c_1,c_2\ge 0$. Furthermore,
\[Rv_2\cdot v_1=(v_1\cdot a_2)(v_2\cdot a_1)-(v_1\cdot a_1)(v_2\cdot a_2)\ge 0,\ Rv_1\cdot v_2=-(v_1\cdot Rv_2)\le 0,\]
and thus, if $t\ge 0$,
\[\dot{y}(t;x,v)\cdot v_1=(c_1CRv_1+c_2CRv_2)\cdot v_1=c_2C(Rv_2\cdot v_1)\ge 0,\]
\[\dot{y}(t;x,v)\cdot v_2=(c_1CRv_1+c_2CRv_2)\cdot v_2=c_1C(Rv_1\cdot v_2)\le 0.\]
Hence, if $t>0$, then $y(t;x,v)\in \Delta$ if and only if $y(t;x,v)\cdot Rv\le 0$. Because of Fact \ref{FACT1} and compactness of $\Delta$, there exists $t^*>0$ such that $y(t^*;x,v)\notin \Delta$. Therefore, $w^*\cdot y(t^*;x,v)>0$. By the intermediate value theorem, there uniquely exists $t^+$ such that $w^*\cdot y(t^+;x,v)=0$.\footnote{Note that, this implies that $t^+$ is the unique $t$ such that $y(t;x,v)$ is proportional to $v$.}

Define a function $t:\Omega^2\to \mathbb{R}_+$ as
\[t(x,v)=\min\{t\ge 0|w^*\cdot y(t;x,v)\ge 0\}.\]
Because of the above consideration, $t(x,v)$ is well-defined if $x$ is not proportional to $v$. If $x$ is proportional to $v$, then the right-hand side of (\ref{IC}) is always $0$. This implies that $y(t;x,v)\equiv x$, and thus $t(x,v)=0$. In both cases, by v) of Lemma \ref{LEM1}, $y(t(x,v);x,v)\in [y_1,y_2]$. Define
\[u^g(x,v)=\frac{\|y(t(x,v);x,v)\|}{\|v\|},\]
\[\succsim^g=(u^g)^{-1}([1,+\infty[).\]
We show the following result.

\begin{lem}\label{LEM2}
Suppose that $g:\Omega\to \mathbb{R}^n_+\setminus \{0\}$ is locally Lipschitz. Then, $u^g$ is a continuous function, and $\succsim^g$ is a complete, p-transitive, continuous, and monotone binary relation on $\Omega$. Moreover, $u^g(x,v)>1$ if and only if $x\succ^gv$, and for $y,z,v\in \Omega$ such that $y,z,v\in V$ for some linear space $V$ with $\dim V\le 2$, $y\succsim^gz$ if and only if $u^g(y,v)\ge u^g(z,v)$. Furthermore, if $x,v\in \Omega$ and $x$ is not proportional to $v$, then $u^g$ is locally Lipschitz around $(x,v)$.
\end{lem}

\begin{proof}
We separate the proof into nine steps.

\begin{step}
$t(x,v)$ is continuous at $(x,v)$ when $x$ is not proportional to $v$.
\end{step}

\begin{proof}[{\bf Proof of Step 1}]
Choose any $\varepsilon>0$. Let $\delta>0$ be so small that $\delta\le \varepsilon$ and $t\mapsto y(t;x,v)$ is defined at $t(x,v)\pm \delta$. By the definition of $t(x,v)$,
\[w^*(x,v)\cdot y(t(x,v)-\delta;x,v)<0,\ w^*(x,v)\cdot y(t(x,v)+\delta;x,v)>0.\]
Therefore, there exists an open neighborhood $U$ of $(x,v)$ such that if $(x',v')\in U$, then
\[w^*(x',v')\cdot y(t(x,v)-\delta;x',v')<0<w^*(x',v')\cdot y(t(x,v)+\delta;x',v').\]
This implies that $|t(x,v)-t(x',v')|<\delta\le \varepsilon$ for any $(x',v')\in U$, as desired. This completes the proof of Step 1. $\blacksquare$
\end{proof}

\begin{step}
$u^g$ is continuous on $\Omega^2$.
\end{step}

\begin{proof}[{\bf Proof of Step 2}]
Because of Step 1, if $x$ is not proportional to $v$, then $u^g(x,v)$ is continuous at $(x,v)$. Therefore, we assume that $x$ is proportional to $v$. Then, $y_1(x,v)=y_2(x,v)=x$. Choose any $\varepsilon>0$. Because of iv) of Lemma \ref{LEM1}, $y_1$ and $y_2$ are continuous around $(x,v)$. Therefore, there exists an open neighborhood $U$ of $(x,v)$ such that if $(x',v')\in U$, then
\[\left|\frac{\|y_i(x',v')\|}{\|v'\|}-\frac{\|x\|}{\|v\|}\right|<\varepsilon.\]
Because $y(t(x',v');x',v')\in [y_1(x',v'),y_2(x',v')]$, for any $(x',v')\in U$,
\[|u^g(x',v')-u^g(x,v)|=\left|\frac{\|y(t(x',v');x,v)\|}{\|v'\|}-\frac{\|x\|}{\|v\|}\right|<\varepsilon,\]
as desired. This completes the proof. $\blacksquare$
\end{proof}

By Step 2, we have that $\succsim^g$ is continuous.

\begin{step}
Let $V$ be a two-dimensional subspace of $\mathbb{R}^n$ and suppose that $x,v,y,z\in \Omega\cap V$. Moreover, suppose that $\mbox{span}\{x,v\}=\mbox{span}\{y,z\}=V$, and there exists $t_1,t_2\in \mathbb{R}$ such that $y(t_1;x,v)=y(t_2;y,z)$. Then, the trajectories of $y(\cdot;x,v)$ and $y(\cdot;y,z)$ coincide.
\end{step}

\begin{proof}[{\bf Proof of Step 3}]
Set $w=y(t_1;x,v)=y(t_2;y,z)$, and let $P^*=P(x,v)=P(y,z)$ and $R^*=R(x,v)$. Consider the following differential equation:
\begin{equation}\label{IC2}
\dot{y}(t)=R^*P^*g(y(t)),\ y(0)=w.
\end{equation}
Define
\[z_1(t)=y(t_1+t/C(x,v);x,v).\]
Then, we can easily check that $z_1(t)$ is the nonextendable solution to (\ref{IC2}). Next, by ii) of Lemma \ref{LEM1}, $R(y,z)$ is either $R^*$ or $-R^*$. Let $s^*=+1$ if $R(y,z)=R^*$, and $s^*=-1$ if $R(y,z)=-R^*$. Define
\[z_2(t)=y(t_2+s^*t/C(y,z);y,z).\]
Then, we can easily check that $z_2(t)$ is also the nonextendable solution to (\ref{IC2}). Because the nonextendable solution to (\ref{IC2}) is unique, $z_1(t)\equiv z_2(t)$. Because the trajectory of $y(\cdot;x,v)$ is that of $z_1(\cdot)$ and the trajectory of $y(\cdot;y,z)$ is that of $z_2(\cdot)$, these are the same. This completes the proof of Step 3. $\blacksquare$
\end{proof}

\begin{step}
Suppose that $x,v\in \Omega$ and $z\in \mbox{span}\{x,v\}$. Then,
\begin{equation}\label{FE}
u^g(u^g(x,z)z,v)=u^g(x,v).
\end{equation}
\end{step}

\begin{proof}[{\bf Proof of Step 4}]
First, suppose that $x$ is proportional to $v$. Then, $z=ax$ and $v=bx$ for some $a,b>0$, and $a=u^g(x,z)$. Thus, $u^g(x,z)z=x$, and thus (\ref{FE}) holds. Hence, hereafter, we assume that $x$ is not proportional to $v$.

Second, suppose that $z$ is not proportional to $v$. By the definition, $y(t(x,z);x,z)$ is proportional to $z$, and thus
\[y(0;u^g(x,z)z,v)=u^g(x,z)z=y(t(x,z);x,z).\]
Moreover,
\[y(0;x,v)=x=y(0;x,z).\]
By Step 3, the trajectories of $y(\cdot;u^g(x,z)z,v)$ and $y(\cdot;x,v)$ coincide,\footnote{Note that, if $x$ is proportional to $z$, then $t(x,z)=0$.} and thus,
\[y(t(u^g(x,z)z,v);u^g(x,z)z,v)=y(t(x,v);x,v),\]
which implies that (\ref{FE}) holds.

Third, supppose that $z$ is proportional to $v$. Recall that
\[y(0;x,v)=x=y(0;x,z).\]
By Step 3, we have that 
\[u^g(x,z)z=y(t(x,z);x,z)=y(t(x,v);x,v),\]
and thus,
\[u^g(x,v)=\frac{\|y(t(x,v);x,v)\|}{\|v\|}=\frac{\|u^g(x,z)z\|}{\|v\|}=u^g(u^g(x,z)z,v).\]
This completes the proof of Step 4. $\blacksquare$
\end{proof}

\begin{step}
$u^g(ax,v)$ is increasing in $a>0$
\end{step}

\begin{proof}[{\bf Proof of Step 5}]
If $x$ is proportional to $v$, then $u^g(ax,v)=au^g(x,v)$, and thus the claim of this step is trivial. Suppose that $x$ is not proportional to $y$, and $a>b>0$. Define $x(s)=(1-s)x+sv$, and
\[c(s)=\frac{u^g(ax,x(s))}{u^g(bx,x(s))}.\]
Then, $c(0)=\frac{a}{b}>1$. Suppose that there exists $s\in [0,1]$ such that $c(s)\le 1$. By the intermediate value theorem, there exists $s^*\in [0,1]$ such that $c(s^*)=1$. By Step 3, $c(s)\equiv 1$, which is a contradiction. Therefore, $c(s)>1$ for all $s\in [0,1]$, and in particular, $c(1)>1$. This implies that $u^g(ax,v)>u^g(bx,v)$, as desired. This completes the proof of Step 5. $\blacksquare$
\end{proof}

\begin{step}
$u^g(x,v)>1$ if and only if $x\succ^gv$. Moreover, if $y,z,v\in V$ for some linear space $V$ with $\dim V\le 2$, then $y\succsim^gz$ if and only if $u^g(y,v)\ge u^g(z,v)$.
\end{step}

\begin{proof}[{\bf Proof of Step 6}]
First, we show the latter claim. Suppose that $y,z,v\in V$ for some linear space $V$ with $\dim V\le 2$. By Step 4, $u^g(u^g(y,z)z,v)=u^g(y,v)$, and by Step 5,
\[y\succsim^gz\Leftrightarrow u^g(y,z)\ge 1\Leftrightarrow u^g(u^g(y,z)z,v)\ge u^g(z,v)\Leftrightarrow u^g(y,v)\ge u^g(z,v),\]
as desired.

Now, suppose that $x\succ^gv$. Then, $v\not\succsim^gx$, and thus,
\[1=u^g(v,v)<u^g(x,v).\]
Conversely, suppose that $u^g(x,v)>1$. By Steps 4 and 5,
\[1=u^g(x,x)=u^g(u^g(x,v)v,x)>u^g(v,x).\]
This implies that $x\succ^gv$. This completes the proof of Step 6. $\blacksquare$
\end{proof}

\begin{step}
$\succsim^g$ is complete and p-transitive.
\end{step}

\begin{proof}[{\bf Proof of Step 7}]
First, choose $x,y\in \Omega$. By Step 6,
\[x\not\succsim^gy\Rightarrow u^g(x,y)<1=u^g(y,y)\Rightarrow y\succsim^gx,\]
which implies that $\succsim^g$ is complete.

Next, choose $x,y,z\in \Omega$ such that $\dim(\mbox{span}\{x,y,z\})\le 2$, $x\succsim^gy$, and $y\succsim^gz$. Then, $u^g(x,y)\ge 1$ and $u^g(y,z)\ge 1$. By Steps 4 and 5,
\[u^g(x,z)=u^g(u^g(x,y)y,z)\ge u^g(y,z)\ge 1,\]
and thus $x\succsim^gz$. Therefore, $\succsim^g$ is p-transitive. This completes the proof of Step 7. $\blacksquare$
\end{proof}

\begin{step}
$\succsim^g$ is monotone.
\end{step}

\begin{proof}[{\bf Proof of Step 8}]
Suppose that $x,v\in \Omega^2$ and $v\gg x$. If $x$ is proportional to $v$, then
\[u^g(v,x)=\frac{\|v\|}{\|x\|}>1,\]
and thus $v\succ^gx$. Hence, we hereafter assume that $x$ is not proportional to $v$. Because $\succsim^g$ is complete, it suffices to show that $u^g(x,v)<1$. Define $z=v-x$. Then, $z\in \mbox{span}\{x,v\}$ and $z\gg 0$. By i), ii), and vi) of Lemma \ref{LEM1},
\[\dot{y}(t;x,v)\cdot Rz=CRPg(y(t;x,v))\cdot Rz=C(g(y(t;x,v))\cdot z)>0,\]
and thus,
\begin{align*}
u^g(x,v)v\cdot Rz=&~y(t(x,v);x,v)\cdot Rz\\
>&~y(0;x,v)\cdot Rz=x\cdot Rz\\
=&~(v-z)\cdot Rz=v\cdot Rz.
\end{align*}
Because $x\cdot a_2=0$,
\begin{align*}
v\cdot Rz=&~v\cdot [(z\cdot a_1)a_2-(z\cdot a_2)a_1]\\
=&~(z\cdot a_1)(v\cdot a_2)-(v\cdot a_2)(v\cdot a_1)\\
=&~-(v\cdot a_2)(x\cdot a_1)=-\|x\|(v\cdot a_2).
\end{align*}
Now,
\[v\cdot a_2=\frac{1}{\|v-(v\cdot a_1)a_1\|}[\|v\|^2-(v\cdot a_1)^2]>0\]
by the Cauchy--Schwarz inequality. Therefore, $v\cdot Rz<0$. This implies that $u^g(x,v)<1$, as desired. This completes the proof of Step 8. $\blacksquare$
\end{proof}

\begin{step}
If $x,v\in \Omega$ and $x$ is not proportional to $v$, then $u^g$ is locally Lipschitz around $(x,v)$.
\end{step}

\begin{proof}[{\bf Proof of Step 9}]
It suffices to show that there exist a neighborhood $U$ of $(x,v)$ and $L>0$ such that if $(y,z),(y',z')\in U$, then
\[|t(y,z)-t(y',z')|\le L\|(y,z)-(y',z')\|.\]
Choose $h>0$ such that $y(\cdot;x,v)$ is defined on $[0,t(x,v)+h]$. Choose any $\varepsilon>0$ and define $C_{\varepsilon}=\{(y,z)\in \Omega|\|(x,v)-(y,z)\|\le \varepsilon\}$. We choose $\varepsilon>0$ so small that $C_{\varepsilon}\subset \Omega$, and for all $(y,z)\in C_{\varepsilon}$, $y$ is not proportional to $z$, $y(\cdot;y,z)$ is defined on $[0,t(x,v)+h]$, and $t(y,z)<t(x,v)+h$. Define
\[m=\inf\{w^*(y,z)\cdot C(y,z)R(y,z)P(y,z)g(w)|(y,z)\in C_{\varepsilon},\ w\in \Delta(y,z)\}.\]
Because $C_{\varepsilon}$ is compact, we have that $m>0$. Because $w^*(y,z)\cdot y(t;y,z)$ is locally Lipschitz and $[0,t(x,v)+h]\times C_{\varepsilon}$ is compact, there exists $L'>0$ such that if $(t,y,z),(t',y',z')\in [0,t(x,v)+h]\times C_{\varepsilon}$, then
\[|w^*(y,z)\cdot y(t;y,z)-w^*(y',z')\cdot y(t';y',z')|\le L'\|(t,y,z)-(t',y',z')\|.\]
Choose $(y,z),(y',z')\in C_{\varepsilon}$. Without loss of generality, we assume that $t(y,z)\le t(y',z')$. By assumption,
\[0\ge w^*(y',z')\cdot y(t(y,z);y',z')\ge -L'\|(y,z)-(y',z')\|.\]
Note that, $w^*(y',z')\cdot y(t(y',z');y',z')=0$, and if $t\in [t(y,z),t(y',z')]$, then $y(t;y',z')\in \Delta(y',z')=\mbox{co}\{y',y_1(y',z'),y_2(y',z')\}$. This implies that
\[w^*(y',z')\cdot \dot{y}(t;y',z')\ge m,\]
and thus,
\[t(y',z')-t(y,z)\le \frac{L'}{m}\|(y,z)-(y',z')\|,\]
as desired. This completes the proof of Step 9. $\blacksquare$
\end{proof}

Steps 2, 6-9 state that all claims of Lemma \ref{LEM2} are correct. This completes the proof. $\blacksquare$
\end{proof}

\setcounter{step}{0}

\begin{lem}\label{LEM3}
Suppose that $g:\Omega\to \mathbb{R}^n_+\setminus \{0\}$ is a locally Lipschitz function. Then, the following two statements are equivalent.
\begin{enumerate}[(i)]
\item $g$ satisfies the weak weak axiom.

\item $\succsim^g$ is weakly convex.

\item $f^g=f^{\succsim^g}$.
\end{enumerate}
\end{lem}

\begin{proof}
We separate the proof into six steps.

\begin{step}
(iii) implies (i).
\end{step}

\begin{proof}[{\bf Proof of Step 1}]
Suppose that (iii) holds, and $g(x)\cdot y\le g(x)\cdot x$. Then, $x\succsim^gy$, and thus $u^g(x,y)\ge 1$. If $g(y)\cdot x<g(y)\cdot y$, then $g(y)\cdot ax<g(y)\cdot y$ for some $a>1$. Because of Step 5 of the proof of Lemma \ref{LEM2}, we have that $u^g(ax,y)>1$, which implies that $ax\succ^gy$. Therefore, $y\in f^g(g(y),g(y)\cdot y)\setminus f^{\succsim^g}(g(y),g(y)\cdot y)$, which contradicts (iii). Therefore, we have that $g(y)\cdot x\ge g(y)\cdot y$, and thus (i) holds. This completes the proof of Step 1. $\blacksquare$
\end{proof}

\begin{step}
$f^{\succsim^g}(p,m)\subset f^g(p,m)$ for all $(p,m)\in (\mathbb{R}^n_+\setminus \{0\})\times \mathbb{R}_{++}$.
\end{step}

\begin{proof}[{\bf Proof of Step 2}]
Choose any $(p,m)\in (\mathbb{R}^n_+\setminus \{0\})\times \mathbb{R}_{++}$. Suppose that $x\in f^{\succsim^g}(p,m)$. Because $\succsim^g$ is monotone, $p\cdot x=m$. Let $v\in \mathbb{R}^n$, $v\neq 0$, $x+v\in \Omega$, and $v\cdot g(x)=0$. Suppose that $p\cdot y(t^*;x,x+v)<m$ for some $t^*\in \mathbb{R}$. Define $z=y(t^*;x,x+v)$. Choose $a>1$ such that $p\cdot az<m$. Then, $az\succ^gz$ and $z\sim^gx$, and thus, $az\succ^gx$. Therefore, $x\notin f^{\succsim^g}(p,m)$, which is a contradiction. Therefore, $p\cdot y(t;x,x+v)\ge m$ for all $t$, and thus, $p\cdot \dot{y}(0;x,x+v)=0$. By vi) of Lemma \ref{LEM1},
\[0=p\cdot \dot{y}(0;x,x+v)=p\cdot CRPg(x),\]
and by ii) of Lemma \ref{LEM1}, $RPg(x)$ is proportional to $v$. Therefore, $p\cdot v=0$. This implies that $g(x)$ is proportional to $p$, and thus $x\in f^g(p,m)$.
\end{proof}

\begin{step}
If (ii) holds, then $f^g(p,m)\subset f^{\succsim}(p,m)$ for all $(p,m)\in (\mathbb{R}^n_+\setminus \{0\})\times \mathbb{R}_{++}$.
\end{step}

\begin{proof}[{\bf Proof of Step 3}]
Suppose that $x\in f^g(p,m)$, and there exists $v\in \Omega$ such that $p\cdot v\le m$ and $v\succ^g x$. Then, $u^g(x,v)<1$, and thus $p\cdot y(t(x,v);x,v)<p\cdot v\le m$. Clearly, $v$ is not proportional to $x$. Let $z=y(t(x,v);x,v)$ and $x(s)=(1-s)x+sz$. Because $\succsim^g$ is weakly convex, $x(s)\succsim^gx$. Recall that $w^*(x,v)=(v\cdot v)x-(v\cdot x)v$. Note that $w^*(x,v)\cdot x(s)=(1-s)w^*(x,v)\cdot x$. Because $\dot{y}(t;x,v)\cdot w^*(x,v)>0$ for all $t$, by the implicit function theorem, we have that there exists an increasing and continuously differentiable function $t^*:[0,1]\to [0,t(x,v)]$ such that $y(t^*(s);x,v)-x(s)$ is proportional to $v$. Because $g(x)$ is proportional to $p$, $p\cdot \dot{y}(0;x,v)=0>p\cdot (z-x)$, and thus,
\[p\cdot \frac{d}{ds}y(t^*(s);x,v)>p\cdot (z-x),\]
for any sufficiently small $s>0$. Therefore, there exists $s\in ]0,1]$ such that $p\cdot y(t^*(s);x,v)>p\cdot x(s)$, which implies that $y(t^*(s);x,v)\gg x(s)$. Because $\succsim^g$ is monotone, we have that $y(t^*(s);x,v)\succ^gx(s)$. By p-transitivity, we have that $x\succ^g x(s)$, which is a contradiction. This completes the proof of Step 3. $\blacksquare$
\end{proof}

Steps 2 and 3 show that (ii) implies (iii). Thus, it suffices to show that (i) implies (ii).

\begin{step}
Suppose that $x,v\in \Omega$ and $x$ is not proportional to $v$. Then, $R(v,x)=-R(x,v)$.
\end{step}

\begin{proof}[{\bf Proof of Step 4}]
Because of ii) of Lemma \ref{LEM1}, we have that either $R(v,x)=R(x,v)$ or $R(v,x)=-R(x,v)$. Therefore, it suffices to show that $R(v,x)v\neq R(x,v)v$. In fact,
\begin{align*}
x\cdot R(x,v)v=&~(a_1(x,v)\cdot v)(a_2(x,v)\cdot x)-(a_2(x,v)\cdot v)(a_1(x,v)\cdot x)\\
=&~-\frac{\|x\|^2\|v\|^2-(x\cdot v)^2}{\|x\|\|v-(a_1(x,v)\cdot v)a_1(x,v)\|}<0,\\
x\cdot R(v,x)v=&~(a_1(v,x)\cdot v)(a_2(v,x)\cdot x)-(a_2(v,x)\cdot v)(a_1(v,x)\cdot x)\\
=&~\frac{\|x\|^2\|v\|^2-(x\cdot v)^2}{\|v\|\|x-(a_1(v,x)\cdot x)a_1(v,x)\|}>0,
\end{align*}
which completes the proof of Step 4. $\blacksquare$
\end{proof}

\begin{step}
Suppose that (i) holds, and choose any $x,v\in \Omega$ such that $x$ is not proportional to $v$ and $x\sim^gv$. For any $s\in [0,1]$, define $x(s)=(1-s)x+sv$. Then, $x(s)\succsim^gx$.
\end{step}

\begin{proof}[{\bf Proof of Step 5}]
Recall that $y(t;x,v)$ is the solution function of the following differential equation:
\[\dot{y}(t)=CRPg(y(t)),\ y(0)=x,\]
where $P=P(x,v), R=R(x,v), C=C(x,v)$. Choose $k\in \mathbb{N}$, and define
\[h^k=\frac{t(x,v)}{k},\ t_i^k=ih^k,\]
\[x_0^k=x,\ x_{i+1}^k=x_i^k+h^kCRPg(x_i^k),\]
and for any $t\in [t_i^k,t_{i+1}^k]$,
\[x^k(t)=\frac{t-t_i^k}{t_{i+1}^k-t_i^k}x_{i+1}^k+\frac{t_{i+1}^k-t}{t_{i+1}^k-t_i^k}x_i^k.\]
It is known that if $k$ is sufficiently large, then $x_i^k$ is defined and in $\Omega$ for all $i\in \{0,...,k\}$, and moreover, $x^k:[0,t(x,v)]\to \Omega$ is a continuous function that uniformly converges to $y(t;x,v)$ as $k\to \infty$.\footnote{The function $x^k$ is called an {\bf explicit Euler approximation} of the solution $y(\cdot;x,v)$ to (\ref{IC}). See, for example, Theorem 1.1 of Iserles (2009).}

Next, define $h(y)=\frac{1}{\|Pg(y)\|}Pg(y)$. Let $a_1=a_1(x,v)$ and $a_2=a_2(x,v)$. Then, there exists $b^i=(b_1^i,b_2^i)$ such that $\|b^1\|=1$, $b_1^i>0$, and
\[h(x_i^k)=b_1^ia_1+b_2^ia_2.\]
Define $v_i^k=x_{i+1}^k-x_i^k$. By definition,
\[v_i^k=\|v_i^k\|(b_1^ia_2-b_2^ia_1).\]
and thus,
\[h(x_i^k)\cdot v_j^k=\|v_j^k\|(b_1^jb_2^i-b_1^ib_2^j).\]
Let $c_i=\frac{b_2^i}{b_1^i}$. We show that $c_i$ is nonincreasing in $i$. By construction, $h(x_i^k)\cdot v_i^k=0$, and by the weak weak axiom, $h(x_{i+1}^k)\cdot v_i^k\le 0$. Therefore, $b_1^ib_2^{i+1}-b_1^{i+1}b_2^i\le 0$, and thus $c_{i+1}\le c_i$, as desired.

Next, for $j\in \{1,...,k\}$, define $v_{0j}^k=x_j^k-x_0^k$. Then, there exists $b^{0j}=(b_1^{0j},b_2^{0j})$ such that
\[v_{0j}^k=\|v_{0j}^k\|(b_1^{0j}a_2-b_2^{0j}a_1).\]
By definition, $b_1^{0j}>0$ and $\|b^{0j}\|=1$. Let $c_{0j}=\frac{b_2^{0j}}{b_1^{0j}}$. Using mathematical induction on $j$, we can show that $c_{0j}$ is nonincreasing.

Now, choose any $s\in [0,1]$, and define $x^k_*(s)=(1-s)x_0^k+sx_k^k$. We show that there uniquely exists $t^k\in [0,t(x,v)]$ such that $x^k_*(s)-x^k(t^k)$ is proportional to $v$. Actually, both $x^k(\cdot)\cdot w^*(x,v)$ and $x^k_*(\cdot)\cdot w^*(x,v)$ are increasing, the intermediate value theorem implies that our claim holds. By Step 4, $R(v,x)=-R$, and thus both $x^k(\cdot)\cdot w^*(v,x)$ and $x^k_*(\cdot)\cdot w^*(v,x)$ are decreasing. This implies that for each $i\in \{1,...,k-1\}$, there exists $s_i\in [0,1]$ such that $x_i^k-x^k_*(s_i)$ is proportional to $x$. Define
\[D(y)=(y\cdot a_1)(v_{0i}^k\cdot a_2)-(y\cdot a_2)(v_{0i}^k\cdot a_1).\]
Then, $D(x)=\|x\|b_1^{0i}>0$. Moreover,
\[D(x^k_*(s_i)-x_i^k)=D(s_iv_{0k}^k-v_{0i}^k)=s_iD(v_{0k}^k)=s_i\|v_{0i}^k\|\|v_{0k}^k\|b_1^{0k}b_1^{0i}[c_{0i}-c_{0k}]\ge 0,\]
which implies that $x^k_*(s_i)\ge x_i^k$. Therefore, $x^k_*(s)\ge x^k(t_k)$. Because $t^k\in [0,t(x,v)]$ for all $k$, there exists a convergent subsequence $(t^{\ell(k)})$. Suppose that $t^*=\lim_{k\to \infty}t^{\ell(k)}$. By the above arguments, $x(s)-y(t^*;x,v)$ is proportional to $v$, and $x(s)\ge y(t^*;x,v)$. Because $\succsim^g$ is monotone, $x(s)\succsim^gy(t^*;x,v)$. Because $x\sim^gy(t^*;x,v)$, we have that $x(s)\succsim^gx$, as desired. This completes the proof of Step 5. $\blacksquare$
\end{proof}

\begin{step}
(i) implies (ii).
\end{step}

\begin{proof}[{\bf Proof of Step 6}]
Suppose that $x\succsim^gy$ and $s\in [0,1]$, and let $x(s)=(1-s)x+sy$. If $x$ is proportional to $y$, then $x(s)\ge y$, and thus $x(s)\succsim^gy$. Hence, we assume that $x$ is not proportional to $y$. Because $x\succsim^gy$, $y\not\succ^gx$, and by Lemma \ref{LEM2}, $u^g(y,x)\le 1$. Define $z(s)=(1-s)u^g(y,x)x+sy$. Because $u^g(y,x)x\sim^gy$, by Step 5, $z(s)\succsim^gy$. Because $x(s)\ge z(s)$, by monotonicity, $x(s)\succsim^gz(s)$. Therefore, by p-transitivity, $x(s)\succsim^gy$. This completes the proof of Step 6. $\blacksquare$
\end{proof}

This completes the proof of Lemma \ref{LEM3}. $\blacksquare$
\end{proof}
\setcounter{step}{0}

\begin{lem}\label{LEM4}
Suppose that $g:\Omega\to \mathbb{R}^n_+\setminus\{0\}$ is locally Lipschitz, and there exist an open set $U\subset \Omega$ and a function $u:U\to \mathbb{R}$ that satisfies property (F). Then, the following claims hold.
\begin{enumerate}[i)]
\item Suppose that $I$ is an open interval and $x:I\to U$ is differentiable, and for all $t\in I$, $g(x(t))\cdot \dot{x}(t)>0$. Then, $u(x(t))$ is increasing.

\item If we can choose $U=\Omega$, then $g$ satisfies Ville's axiom.
\end{enumerate}
\end{lem}

\begin{proof}
First, we show i). Suppose that $I$ is an open interval and $x:I\to \Omega$ is differentiable, and for all $t\in I$, $g(x(t))\cdot \dot{x}(t)>0$. Define $h(t)=u(x(t))$. Choose any $t^*\in I$. We first show that there exists $\varepsilon>0$ such that if $0<|t-t^*|<\varepsilon$, then $h(t)>h(t^*)$ if $t>t^*$, and $h(t)<h(t^*)$ if $t<t^*$. Let $X=u^{-1}(u(x(t^*)))$. By assumption, $X$ is an $n-1$ dimensional $C^1$ manifold, and $g(x)$ is orthogonal to $T_x(X)$ for each $x\in X$. Choose a corresponding $C^1$ local parametrization $\varphi:U\to V$ of $X$, where $U$ is open neighborhood of $0$ in $\mathbb{R}^{n-1}$, $V$ is an open neighborhood of $x(t^*)$ in $X$, and $\varphi(0)=x(t^*)$. Define $\Phi(z,a)=\varphi(z)+ag(x(t^*))$. Then, $D\Phi(0,0)$ is regular, and by the inverse function theorem, there exists an open and convex neighborhood $U'\subset \mathbb{R}^n$ of $(0,0)$ and an open neighborhood $V'\subset \mathbb{R}^n$ of $x(t^*)$ such that $\Phi:U'\to V'$ is a $C^1$ diffeomorphism. We can assume that for $(z,h)\in U'$, $u(\Phi(z,h))=u(x(t^*))$ if and only if $h=0$. Choose $h_1>0$ such that $(0,h_1)\in U'$. Then, $\Phi(0,h_1)=x(t^*)+hg(x(t^*))\gneq x(t^*)$. Because $u$ is increasing and $h_1\neq 0$, $u(\Phi(0,h_1))>u(x(t^*))$. By the convexity of $U'$ and the intermediate value theorem, we have that $u(\Phi(z,h))>u(x(t^*))$ if $h>0$. By the symmetrical arguments, we have that $u(\Phi(z,h))<u(x(t^*))$ if $h<0$.

Define $y(t)=\Phi^{-1}(x(t))$. We can assume that there exists $\delta>0$ such that $y(t)$ is defined on $I\equiv ]t^*-\delta,t^*+\delta[$. By assumption, $g(x(t^*))\cdot \dot{x}(t^*)>0$, and thus $\dot{y}_n(t^*)=\|g(x(t^*))\|^{-2}(g(x(t^*))\cdot \dot{x}(t^*))>0$. Therefore, there exists $\varepsilon>0$ such that 1) if $0<t-t^*<\varepsilon$, then $y_n(t)>0$, and 2) if $0<t^*-t<\varepsilon$, then $y_n(t)<0$. Hence, when $0<|t-t^*|<\varepsilon$, then $h(t)>h(t^*)$ if $t>t^*$ and $h(t)<h(t^*)$ if $t<t^*$, as desired. 

Let $t',t''\in I$ and $t'<t''$. Suppose that $h(t'')\le h(t')$. Then, there exists $t^*\in \arg\max\{h(t)|t\in [t',t'']\}$ such that $t^*\neq t''$. By the above arguments, there exists $t'''\in ]t^*,t'']$ such that $h(t''')>h(t^*)$, which is a contradiction. Therefore, $h(t)$ is increasing, and i) holds.

Now, suppose that $U=\Omega$, and $x:[0,T]\to \Omega$ is a Ville curve, and define $h(t)=u(x(t))$. Because $x(t)$ is piecewise $C^1$, there exists a finite increasing sequence $t_0,...,t_m\in [0,T]$ such that $t_0=0, t_m=T$ and the restriction of $x(t)$ on $[t_i,t_{i+1}]$ is continuously differentiable. If $t_i<t'<t''<t_{i+1}$, then by Step 1, $h(t'')>h(t')$. By continuity of $h$, $h$ is increasing in $[t_i,t_{i+1}]$, and therefore $h(t_0)<h(t_1)<...<h(t_m)$. This implies that $u(x(0))<u(x(T))=u(x(0))$, which is a contradiction. Therefore, there is no Ville curve, and ii) holds. This completes the proof. $\blacksquare$
\end{proof}

\begin{lem}\label{LEM5}
Suppose that $g:\Omega\to \mathbb{R}^n_+\setminus \{0\}$ is locally Lipschitz. For any $v\in \Omega$, define $u^g_v(x)=u^g(x,v)$. Then, the following conditions are equivalent.
\begin{enumerate}[1)]
\item $g$ satisfies Ville's axiom.

\item $u^g_v$ satisfies property (F).

\item The equality
\begin{equation}\label{FE2}
u^g(u^g(x,z)z,v)=u^g(x,v)
\end{equation}
holds for all $x,z,v\in \Omega$.

\item $\succsim^g$ is represented by $u^g_v$.

\item $\succsim^g$ is transitive.
\end{enumerate}
Moreover, if $u$ is another function that satisfies property (F), then $u$ is a monotone transform of $u^g_v$. Furthermore, if $g$ is $C^k$, then $u^g_v$ is $C^k$ and $\nabla u^g_v(x)\neq 0$ for all $x\in \Omega$. 
\end{lem}

\begin{proof}
We separate the proof into six steps.

\begin{step}
1) implies 5).
\end{step}

\begin{proof}[{\bf Proof of Step 1}]
We show that if 5) is violated, then 1) is also violated. Suppose that $\succsim^g$ is not transitive. Then, there exist $x,y,z\in \Omega$ such that $x\succsim^gy$, $y\succsim^gz$, and $z\succ^gx$. By p-transitivity, each two vectors of $x,y,z$ are linearly independent. Because $u^g(x,y)\ge 1$, by Lemma \ref{LEM2}, $u^g(u^g(x,y)y,z)\ge u^g(y,z)\ge 1$. By monotonicity and p-transitivity, $u^g(u^g(x,y)y,z)z\succ^gx$, and thus by replacing $y$ with $u^g(x,y)y$ and $z$ with $u^g(u^g(x,y)y,z)z$, we can assume that $x\sim^gy$ and $y\sim^gz$, and $z\succ^gx$. Let $x^*=u^g(z,x)x$. Consider the following differential equation:
\begin{equation}\label{VC1}
\dot{w}(t)=(g(w(t))\cdot x^*)z-(g(w(t))\cdot z)x^*+\varepsilon z,\ w(0)=x^*.
\end{equation}
Let $w_1(t;\varepsilon)$ be the solution function of (\ref{VC1}). Clearly, $w_1(t;0)=y(t;x^*,z)$, and thus $w_1(\cdot;0)$ is defined on $[0,t(x^*,z)+h_1]$ for some $h_1>0$, and $w_1(t(x^*,z);0)=z$. By Fact \ref{FACT2}, there exists $\varepsilon_1>0$ such that if $0<\varepsilon<\varepsilon_1$, then $w_1(\cdot;\varepsilon)$ is defined on $[0,t(x^*,z)+h_1]$, and there exists $t_1(\varepsilon)>0$ such that $v_1(\varepsilon)\equiv w_1(t_1(\varepsilon);\varepsilon)$ is proportional to $z$. Because $\dot{w}_1(0;\varepsilon)\gg \dot{w}_1(0;0)$, $w_1(t;\varepsilon)\gg y(t;x^*,z)$ for sufficiently small $t>0$. Suppose that $w_1(t;\varepsilon)\le y(t;x^*,z)$ for some $t\in [0,t(x^*,z)]$. Let $t^*=\inf\{t'\in [0,t]|w_1(t';\varepsilon)\le y(t';x^*,z)\}$. Then, $t^*>0$. Moreover, $\dot{w}(t^*;\varepsilon)\gg \dot{y}(t^*;x^*,z)$, and thus there exists $t^+<t^*$ such that $w(t^+;\varepsilon)\ll y(t^+;x^*,z)$, which contradicts the definition of $t^*$. Therefore, $w_1(t;\varepsilon)\gg y(t;x^*,z)$ for all $t$, and thus, $v_1(\varepsilon)=a_1(\varepsilon)z$ for some $a_1(\varepsilon)>1$.

Next, consider the following differential equation:
\begin{equation}\label{VC2}
\dot{w}(t)=(g(w(t))\cdot v_1(\varepsilon))y-(g(w(t))\cdot y)v_1(\varepsilon)+\varepsilon y,\ w(0)=v_1(\varepsilon).
\end{equation}
Let $w_2(t;\varepsilon)$ be the solution function of (\ref{VC2}). Clearly, $w_2(t;0)=y(t;z,y)$, and thus $w_2(\cdot;0)$ is defined on $[0,t(z,y)+h_2]$ for some $h_2>0$, and $w_2(t(z,y);0)=y$. Thus, there exists $\varepsilon_2>0$ such that if $0<\varepsilon<\varepsilon_2$, then $w_2(\cdot;\varepsilon)$ is defined on $[0,t(z,y)+h_2]$, and there exists $t_2(\varepsilon)>0$ such that $v_2(\varepsilon)\equiv w_2(t_2(\varepsilon);\varepsilon)$ is proportional to $y$. By almost tha same arguments as in the last paragraph, we can show that $w_2(t;\varepsilon)\gg y(t;z,y)$ for all $t$, and thus, $v_2(\varepsilon)=a_2(\varepsilon)y$ for some $a_2(\varepsilon)>1$.

Third, consider the following differential equation:
\begin{equation}\label{VC3}
\dot{w}(t)=(g(w(t))\cdot v_2(\varepsilon))x-(g(w(t))\cdot x)v_2(\varepsilon)+\varepsilon x,\ w(0)=v_2(\varepsilon).
\end{equation}
Let $w_3(t;\varepsilon)$ be the solution function of (\ref{VC3}). Clearly, $w_3(t;0)=y(t;y,x)$, and thus $w_3(\cdot;0)$ is definedon $[0,t(y,x)+h_3]$ for some $h_3>0$, and $w_3(t(y,x);0)=x$. Thus, there exists $\varepsilon_3>0$ such that if $0<\varepsilon<\varepsilon_3$, then $w_3(\cdot;\varepsilon)$ is defined on $[0,t(y,x)+h_3]$, and there exists $t_3(\varepsilon)>0$ such that $v_3(\varepsilon)\equiv w_3(t_3(\varepsilon);\varepsilon)$ is proportional to $x$. By almost the same arguments as in the last paragraphs, we can show that $w_3(t;\varepsilon)\gg y(t;y,x)$ for all $t$, and thus, $v_3(\varepsilon)=a_3(\varepsilon)x$ for some $a_3(\varepsilon)>1$.

Because $w_1(t;\varepsilon),w_2(t;\varepsilon),w_3(t;\varepsilon)$ are continuous in $\varepsilon$, there exists $\varepsilon^*>0$ such that $v_3(\varepsilon^*)\ll x^*$. Let $t_1^*=t_1(\varepsilon^*)$, $t_2^*=t_1^*+t_2(\varepsilon^*)$, $t_3^*=t_2^*+t_3(\varepsilon^*)$, and $T=t_3^*+1$. Define a piecewise $C^1$ closed curve $x:[0,T]$ as follows:
\[x(t)=\begin{cases}
w_1(t;\varepsilon^*) & \mbox{if }0\le t\le t_1^*,\\
w_2(t-t_1^*;\varepsilon^*) & \mbox{if }t_1^*\le t\le t_2^*,\\
w_3(t-t_2^*;\varepsilon^*) & \mbox{if }t_2^*\le t\le t_3^*,\\
(T-t)v_3(\varepsilon^*)+(t-t_3^*)x^* & \mbox{if }t_3^*\le t\le T.
\end{cases}\]
Then, we can easily check that $x(t)$ is a Ville curve, and thus 1) is violated. This completes the proof of Step 1. $\blacksquare$
\end{proof}

\begin{step}
3), 4), and 5) are mutually equivalent.
\end{step}

\begin{proof}[{\bf Proof of Step 2}]
Clearly, 4) implies 5).

Suppose that 5) holds. Choose any $x,z,v\in \Omega$. Then,
\[u^g(x,v)v\sim^gx\sim^gu^g(x,z)z\sim^gu^g(u^g(x,z)z,v)v,\]
which implies that (\ref{FE2}) is correct. Therefore, 3) holds.

Finally, suppose that 3) holds. Choose any $v\in \Omega$ and define $u^g_v(x)=u^g(x,v)$. By Lemma \ref{LEM2}, $u^g(az,v)$ is increasing in $a$. Therefore, if $x\succsim^gz$, then
\[u^g_v(x)=u^g(x,v)=u^g(u^g(x,z)z,v)\ge u^g(z,v)=u^g_v(z),\]
and if $u^g_v(x)\ge u^g_v(z)$, then
\[u^g(z,v)=u^g_v(z)\le u^g_v(x)=u^g(x,v)=u^g(u^g(x,z)z,v),\]
which implies that $u^g(x,z)\ge 1$ and $x\succsim^gz$. hence, 4) holds. This completes the proof of Step 2. $\blacksquare$
\end{proof}

\begin{step}
2) implies 1).
\end{step}

\begin{proof}[{\bf Proof of Step 3}]
This result immediately follows from Lemma \ref{LEM4}. $\blacksquare$
\end{proof}

\begin{step}
5) implies 1).
\end{step}

\begin{proof}[{\bf Proof of Step 4}]
Suppose that 5) holds. By Step 2, 3) and 4) also hold. First, choose any $x,v\in \Omega$, and choose $z\in \Omega$ such that both $x,v$ are not proportional to $z$. By (\ref{FE2}),
\[u^g_v(y)=u^g_v(u^g(y,z)z)\]
for any $y$ that is sufficiently near to $x$, and by Lemma \ref{LEM2}, the right-hand side is locally Lipschitz around $x$. By Fact \ref{FACT7}, $u^g_v$ is locally Lipschitz on $\Omega$.

Next, choose any $x,z\in \Omega$ such that $z\gg x$. Because $\succsim^g$ is monotone, $z\succ^gx$, and by 4), $u^g_v(z)>u^g_v(x)$. Therefore, $u^g_v$ is increasing.

Suppose that $n=2$. Choose any $a>0$ and let $X\equiv (u^g_v)^{-1}(a)$. Because $av\in X$, $X\neq \emptyset$. Choose $x\in X$ that is not proportional to $v$. Then, $X$ coincides the trajectory of $y(\cdot;x,v)$, and thus it is $1$ dimensional $C^1$ manifold. Moreover, if $z\in X$, then $z=y(t;x,v)$ for some $t$, and $\dot{y}(t;x,v)\cdot g(z)=0$. Clearly, $T_z(X)=\mbox{span}\{\dot{y}(t;x,v)\}$, and thus our claim is correct.

Hence, we hereafter assume that $n\ge 3$. Because both $u^g_v$ and $g$ are locally Lipschitz, by Rademacher's theorem, these are differentiable almost everywhere. Suppose that $u^g_v$ is differentiable at $x$. Choose linearly independent vectors $v_1,...,v_{n-1}\in \mathbb{R}^n$ such that $g(x)\cdot v_i=0$ and $x+v_i\in \Omega$ for any $i\in \{1,...,n-1\}$. Define $w_i=\dot{y}(0;x,x+v_i)$. Then, we must have that $w_i$ is positively proportional to $v_i$. Because $Du^g_v(x)w_i=0$, we have that $\nabla u^g_v(x)=\lambda(x)g(x)$ for some $\lambda(x)\ge 0$. If $x$ is proportional to $v$, then $u^g_v(ax)=au^g_v(x)$, and thus $\lambda(x)\neq 0$. Suppose that $x$ is not proportional to $v$ and $\lambda(x)=0$. Note that,
\[a=u^g(ax,x)=u^g(u^g_v(ax)v,x)\]
and $u^g_x$ is locally Lipschitz around $u^g_v(x)v$. Therefore, there exists $L>0$ such that if $a>1$ and $a-1$ is sufficiently small, then $0<u^g(u^g_v(ax)v,x)-u^g(u^g_v(x)v,x)\le L(u^g_v(ax)-u^g_v(x))\|v\|$. This implies that,
\[1=\lim_{a\downarrow 1}\frac{u^g(u^g_v(ax)v,x)-u^g(u^g_v(x)v,x)}{a-1}\le \lim_{a\downarrow 1}\frac{(u^g_v(ax)-u^g_v(x))L\|v\|}{a-1}=0,\]
which is a contradiction. Therefore, we must have that $\lambda(x)>0$.

Choose any $a>0$ and define $X=(u^g_v)^{-1}(a)$. Because $av\in X$, $X$ is a nonempty set. Choose any $x^*\in X$. Without loss of generality, we assume that $g_n(x^*)\neq 0$. Therefore, there exists $\varepsilon>0$ such that if $|x_i-x_i^*|<\varepsilon$ for all $i\in \{1,...,n\}$, then $x\in \Omega$ and $g_n(x)\neq 0$. If $g_n(x)\neq 0$, then define $f_i(x)=-\frac{g_i(x)}{g_n(x)}$ for each $i\in \{1,...,n-1\}$. Hereafter, we use the following notation: if $x=(x_1,...,x_n)$, then $\tilde{x}=(x_1,...,x_{n-1})$. Consider the following differential equation.
\begin{equation}\label{SH2}
\dot{c}(t)=f((1-t)\tilde{x}^*+t\tilde{x},c(t))\cdot (\tilde{x}-\tilde{x}^*),\ c(0)=x_n^*+h.
\end{equation}
Let $c(t;\tilde{x},h)$ be the solution function of (\ref{SH2}). Because $c(t;\tilde{x}^*,0)\equiv x_n^*$, it is defined on $[0,1]$. By Fact \ref{FACT2}, there exists $\delta>0$ such that if $|x_i-x_i^*|\le \delta$ for each $i$ and $|h|\le \delta$, then $c(\cdot;\tilde{x},h)$ is defined on $[0,1]$ and $|c(t;\tilde{x},h)-x_n^*|<\varepsilon$ for all $t\in [0,1]$. It is clear that $\delta<\varepsilon$. Let $\bar{U}=\{(\tilde{x},h)||x_i-x_i^*|\le \delta,\ |h|\le \delta\}$, and define $E(\tilde{x},h)=c(1;\tilde{x},h)$ for $(\tilde{x},h)\in \bar{U}$. If $\tilde{x}=\tilde{x}^*$, then $E(\tilde{x},h)=x_n^*+h$, and thus $u^g_v(\tilde{x},E(\tilde{x},h))=u^g_v(x^*+he_n)$. Suppose that $\tilde{x}\neq \tilde{x}^*$. Then, there exists $i^*\in \{1,...,n-1\}$ such that $x_{i^*}\neq x_{i^*}^*$. By Fact \ref{FACT9} and Fubini's theorem, there exists a sequence $(x^k,h^k)$ such that $(x^k,h^k)\to (x,h)$ as $k\to \infty$, and $u^g_v$ is differentiable at $((1-t)\tilde{x}^*+t\tilde{x}^k,c(t;\tilde{x}^k,h^k))$ for almost all $t\in [0,1]$. Therefore,
\[\frac{d}{dt}u^g_v((1-t)\tilde{x}^*+t\tilde{x}^k,c(t;\tilde{x}^k,h^k))=0\]
almost everywhere, which implies that $u^g_v(\tilde{x}^k,E(\tilde{x}^k,h^k))=u^g_v(x^*+h^ke_n)$. Letting $k\to \infty$, we obtain that $u^g_v(\tilde{x},E(\tilde{x},h))=u^g_v(x^*+he_n)$.

Suppose that $h_1<h_2$ and $E(\tilde{x},h_1)\ge E(\tilde{x},h_2)$. Then, there exists $t^*\in [0,1]$ such that $c(t^*;\tilde{x},h_1)=c(t^*;\tilde{x},h_2)$, and by Fact \ref{FACT1}, $c(0;\tilde{x},h_1)=c(0;\tilde{x},h_2)$, which is a contradiction. Therefore, $E(\tilde{x},h)$ is increasing in $h$. Consider the following differential equation:
\[\dot{d}(t)=f((1-t)\tilde{x}^*+t\tilde{x},d(t))\cdot (\tilde{x}-\tilde{x}^*),\ d(1)=x_n^*+h'.\]
Let $d(t;\tilde{x},h')$ be the solution function of this equation. Because $d(t;\tilde{x}^*,E(\tilde{x}^*,h)-x_n^*)=c(t;\tilde{x}^*,h)=x_n^*+h$, there exists $\delta'>0$ such that if $|x_i-x_i^*|\le \delta'$ for each $i$ and $|h|\le \delta'$, then $|E(\tilde{x},h)-x_n^*|\le \delta$ and $t\mapsto d(t;\tilde{x},E(\tilde{x},h)-x_n^*)$ is defined on $[0,1]$. In this case, $c(t;\tilde{x}^*,h)=d(t;\tilde{x},E(\tilde{x},h)-x_n^*)$, and thus $d(0;\tilde{x},E(\tilde{x},h)-x_n^*)=x_n^*+h$. By repeating the above argument, we can show that $d(0;\tilde{x},h')$ is increasing in $h'$. Without loss of generality, we can assume that $\delta'<\delta$, and $d(t;\tilde{x},h')$ is defined whenever $t\in [0,1]$, $|x_i-x_i^*|\le \delta'$, and $|h'|\le \delta$. Define
\[F(\tilde{x},h')=d(0;\tilde{x},h').\]
Then, if $|x_i-x_i^*|\le \delta'$ for all $i$ and $|h|\le \delta'$,
\[E(\tilde{x},h)=x_n^*+h'\Leftrightarrow F(\tilde{x},h')=x_n^*+h.\]
Define $U=\{(\tilde{x},h)||x_i-x_i^*|\le \delta',\ |h|\le \delta'\}$, and let $W$ be the set of all $(\tilde{x},h)\in U$ such that $E$ is differentiable at $(\tilde{x},h)$. By Rademacher's theorem, $U\setminus W$ is a null set. Let $V=\{(\tilde{x},h')||x_i-x_i^*|\le \delta',\ |h'|\le \delta\}$, and define $\Phi(\tilde{x},h')=(\tilde{x},F(\tilde{x},h')-x_n^*)$. Then, $U\subset \Phi(V)$. Let $W'$ be the set of all $(\tilde{x},h')\in V$ such that $u^g_v$ is differentiable at $(\tilde{x},x_n^*+h')$. Then, $V\setminus W'$ is a null set. Because $\Phi$ is Lipschitz on $V$, $\Phi(V\setminus W')$ is also a null set. Therefore, $U\setminus (W\cap \Phi(W'))$ is a null set. If $(\tilde{x},h)\in W\cap \Phi(W')$, then $E$ is differentiable at $(\tilde{x},h)$ and $u^g_v$ is differentiable at $(\tilde{x},E(\tilde{x},h))$. Choose any $(\tilde{x},h)\in W\cap \Phi(W')$. Then,
\begin{align*}
0=&~\left.\frac{\partial}{\partial x_i}u^g_v(\tilde{y},E(\tilde{y},h))\right|_{\tilde{y}=\tilde{x}}\\
=&~\frac{\partial u^g_v}{\partial x_i}(\tilde{x},E(\tilde{x},h))+\frac{\partial u^g_v}{\partial x_n}(\tilde{x},E(\tilde{x},h))\frac{\partial E}{\partial x_i}(\tilde{x},h)\\
=&~\lambda(\tilde{x},E(\tilde{x},h))\left[g_i(\tilde{x},E(\tilde{x},h))+g_n(\tilde{x},E(\tilde{x},h))\frac{\partial E}{\partial x_i}(\tilde{x},h)\right],
\end{align*}
which implies that
\[\frac{\partial E}{\partial x_i}(\tilde{x},h)=f_i(\tilde{x},E(\tilde{x},h)).\]
Now, choose any $\tilde{x}$ such that $|x_i-x_i^*|<\delta'$ for each $i\in \{1,...,n-1\}$. By Fubini's theorem, there exists $\delta''>0$ and a sequence $(\tilde{x}^k,h^k)$ on $U$ such that $(\tilde{x}^k,h^k)\to (\tilde{x},0)$ as $k\to \infty$, and for all $k$ and almost all $t\in [-\delta'',\delta'']$, $(\tilde{x}^k+te_i,h^k)\in W\cap \Phi(W')$. This implies that,
\[E(\tilde{x}^k+te_i,h^k)-E(\tilde{x}^k,h^k)=\int_0^tf_i(\tilde{x}^k+se_i,E(\tilde{x}^k+se_i,h^k))ds.\]
By the dominated convergence theorem,
\[E(\tilde{x}+te_i,0)-E(\tilde{x},0)=\int_0^tf_i(\tilde{x}+se_i,E(\tilde{x}+se_i,0))ds,\]
and thus,
\[\frac{\partial E}{\partial x_i}(\tilde{x},0)=f_i(\tilde{x},E(\tilde{x},0)).\]
This implies that $\tilde{x}\mapsto E(\tilde{x},0)$ is continuously differentiable around $\tilde{x}^*$, and thus the mapping
\[\Psi:\tilde{x}\mapsto (\tilde{x},E(\tilde{x},0))\]
is a $C^1$ local parametrization of $X$ around $x^*$. Therefore, $X$ is an $n-1$ dimensional $C^1$ manifold. Moreover, because $g(x^*)\cdot \frac{\partial \Psi}{\partial x_i}(\tilde{x}^*)=0$ for all $i\in \{1,...,n-1\}$, we have that $g(x^*)\cdot v=0$ for all $v\in T_{x^*}(X)$, as desired. This completes the proof of Step 4. $\blacksquare$
\end{proof}

\begin{step}
If $g$ is $C^k$ and satisfies Ville's axiom, then $u^g_v$ is $C^k$ and there exists a continuous function $\lambda:\Omega\to \mathbb{R}_{++}$ such that $\nabla u^g_v(x)=\lambda(x)g(x)$ for all $x\in \Omega$.
\end{step}

\begin{proof}[{\bf Proof of Step 5}]
In this case, $y(t;x,v)$ is $C^k$ by Fact \ref{FACT3}. Suppose that $(x,v)\in \Omega$. Choose $z\in \Omega$ such that both $x,v$ are not proportional to $z$. By the implicit function theorem, we have that $t(y,w)$ is $C^k$ around $(x,z)$, and thus $u^g$ is $C^k$ around $(x,z)$. By the same argument, we have that $u^g$ is $C^k$ around $(u^g(x,z)z,v)$. By 3),
\[u^g(u^g(y,z)z,w)=u^g_v(y,w),\]
and the left-hand side is $C^k$ around $(x,v)$. Therefore, $u^g$ is $C^k$ and $u^g_v$ is also $C^k$. Moreover,
\[a=u^g_v(av)=u^g_v(u^g(av,x)x),\ u^g(u^g_v(x)v,x)=1.\]
Differentiating the first equality in $a$ at $a=u^g_v(x)$,
\[1=(\nabla u^g_v(x)\cdot x)\times (Du^g_x(u^g_v(x)v)v,\]
which implies that $\nabla u^g_v(x)\neq 0$. Let $x\in \Omega$, and $X=(u^g_v)^{-1}(u^g_v(x))$. Let $g(x)\cdot w=0$. Then, $w\in T_x(X)$. By Fact \ref{FACT4}, there exists an open interval $I$ including $0$ and a $C^1$ function $c:I\to X$ such that $c(0)=x$ and $\dot{c}(0)=w$. Therefore, $\nabla u^g_v(x)\cdot w=0$. Thus, $\nabla u^g_v(x)$ is proportional to $g(x)$. Because $u^g_v$ is increasing, there exists $\lambda(x)>0$ such that $\nabla u^g_v(x)=\lambda(x)g(x)$. This completes the proof of Step 5.  $\blacksquare$
\end{proof}

\begin{step}
If $g$ satisfies Ville's axiom and $u:\Omega\to \mathbb{R}$ is a function that satisfies property (F), then $u$ is a monotone transform of $u^g_v$.
\end{step}

\begin{proof}[{\bf Proof of Step 6}]
Choose any $x\in \Omega$ such that $x$ is not proportional to $v$. We first show that $u(y(t;x,v))=u(x)$ for all $t\in [0,t(x,v)]$.

Consider the following differential equation:
\[\dot{z}(t)=(g(z(t))\cdot x)v-(g(z(t))\cdot v)x+hv,\ z(0)=x.\]
Let $z(t;h)$ be the solution function of the above equation. Because $z(t;0)=y(t;x,v)$, there exists $\varepsilon>0$ such that if $|h|<\varepsilon$, then $z(\cdot;h)$ is defined on $[0,t(x,v)]$. Let $h>0$. Because of i) of Lemma \ref{LEM4}, $u(z(t;h))$ is an increasing function, and thus $u(z(t;h))>u(x)$ for any $t>0$. Letting $h\to 0$, we have that $u(y(t;x,v))\ge u(x)$. By the symmetrical argument, we can show that $u(y(t;x,v))=u(x)$, and thus our claim holds. This implies that,
\[u(x)=u(y(t(x,v);x,v))=u(u^g_v(x)v).\]
Clearly, this equality holds even when $x$ is proportional to $v$.

Define $c(a)=u(av)$. Because $u$ is increasing, $c$ is also increasing. Therefore,
\[c(u^g_v(x))=u(u^g_v(x)v)=u(x),\]
as desired. This completes the proof of Step 6. $\blacksquare$
\end{proof}

Steps 1-6 show that the claims of Lemma \ref{LEM5} are correct. This completes the proof. $\blacksquare$
\end{proof}
\setcounter{step}{0}

\subsection{Proof of Theorem \ref{THM1}}
Let $g:\Omega\to \mathbb{R}^n_+\setminus \{0\}$ be locally Lipschitz. Suppose that $f^g=f^{\succsim}$ for some complete binary relation $\succsim$ on $\Omega$ that satisfies the LNST condition. Suppose that $g(x)\cdot y\le g(x)\cdot x$ and $g(y)\cdot x<g(y)\cdot y$. Because $x\in f^g(g(x),g(x)\cdot x)$, $y\not\succ x$, and because $\succsim$ is complete, $x\succsim y$. By the LNST condition, there exists $z\in \Omega$ such that $g(y)\cdot z<g(y)\cdot y$ and $z\succ y$, which contradicts that $y\in f^g(g(y),g(y)\cdot y)$. Therefore, $g$ satisfies the weak weak axiom.

Conversely, suppose that $g$ satisfies the weak weak axiom. By Lemma \ref{LEM3}, $f^g=f^{\succsim^g}$. By Lemma \ref{LEM2}, $\succsim^g$ is complete. Suppose that $x\succsim^gy$ and $U$ is a neighborhood of $x$. Then, there exists $a>1$ such that $ax\in U$. Because $\succsim^g$ is monotone, $ax\succ^gx$, and by p-transitivity, $ax\succ^gy$. Therefore, $\succsim^g$ satisfies the LNST condition. This completes the proof of I).

Next, suppose that $g$ satisfies Ville's axiom. By Lemma \ref{LEM5}, $u^g_v$ satisfies property (F). Conversely, suppose that there exists a function $u:\Omega\to \mathbb{R}$ that satisfies property (F). By Lemma \ref{LEM4}, $g$ satisfies Ville's axiom. By definition, $u^g_v(av)=a$. If $g$ is $C^k$, by Lemma \ref{LEM5}, $u^g_v$ is also $C^k$, and $\nabla u^g_v(x)=\lambda(x)g(x)$ for all $x\in \Omega$, where $\lambda:\Omega\to \mathbb{R}_{++}$. This completes the proof of II).

Thirdly, suppose that $g$ satisfies both the weak weak axiom and Ville's axiom. By Lemmas \ref{LEM3} and \ref{LEM5}, $u^g_v$ is a quasi-concave function that satisfies property (F). Conversely, suppose that there exists a quasi-concave function $u:\Omega\to \mathbb{R}$ that satisfies property (F). By Lemma \ref{LEM4}, $g$ satisfies Ville's axiom, and by Lemma \ref{LEM5}, $u$ is a monotone transform of $u^g_v$. This implies that $u^g_v$ is quasi-concave, and by Lemma \ref{LEM3}, $g$ satisfies the weak weak axiom. 

Fourth, suppose that $g$ satisfies both the weak weak axiom and Ville's axiom. Then, $\succsim^g$ is a weak order such that $f^g=f^{\succsim^g}$. Conversely, suppose that there exists a weak order $\succsim$ such that $f^g=f^{\succsim}$. We show that $\succsim$ is locally non-satiated. Let $x\in \Omega$ and $U$ is an open neighborhood of $x$. Then, there exists $y\in U$ such that $y\gg x$. Because $y\in f^g(g(y),g(y)\cdot y)$ and $g(y)\cdot x<g(y)\cdot y$, we have that $y\succ x$, as desired. Therefore, $\succsim$ satisfies the LNST condition, and thus $g$ satisfies the weak weak axiom.

Note that, if $p\cdot x<m$, then $x\notin f^g(p,m)$, and thus if $y\in f^g(p,m)$, then $y\succ x$. Choose any $x,v\in \Omega$ such that $x$ is not proportional to $v$. As in the proof of Step 5 of Lemma \ref{LEM3}, choose $k\in \mathbb{N}$, and define
\[h^k=\frac{t(x,v)}k,\ t_i^k=ih^k,\]
\[x_0^k=x,\ x_{i+1}^k=x_i^k+h^kCRPg(x_i^k),\]
and for any $t\in [t_i^k,t_{i+1}^k]$,
\[x^k(t)=\frac{t-t_i^k}{t_{i+1}^k-t_i^k}x_{i+1}^k+\frac{t_{i+1}^k-t}{t_{i+1}^k-t_i^k}x_i^k.\]
As we have mentioned, $x^k(t)$ uniformly converges to $y(t;x,v)$ on $[0,t(x,v)]$ as $k\to \infty$. Because $g(x_i^k)\cdot (x_{i+1}^k-x_i^k)=0$, we have that $x_i^k\succsim x_{i+1}^k$. Moreover, if $t\in [t_i,t_{i+1}]$, then $g(x_i^k)\cdot (x^k(t)-x_i^k)=0$, and thus $x_i^k\succsim x^k(t)$. By transitivity of $\succsim$, we have that $x\succsim x^k(t)$. If $y(t;x,v)\gg z$ for some $t\in [0,t(x,v)]$, then there exists $k$ such that $x^k(t)\gg z$, which implies that $x\succ z$. In conclusion, we obtain the following: if $y(t;x,v)\gg z$ for some $t\in [0,t(x,v)]$, then $x\succ z$.

Suppose that $\succsim^g$ is not transitive. As in the proof of Step 1 of Lemma \ref{LEM5}, there exists $x,y,z\in \Omega$ such that $x\sim^gy$, $y\sim^gz,$ and $z\succ^gx$. Because $z\succ^gx$, $az\succ^gx$ for some $a\in ]0,1[$. Because $y\sim^gz$, $y\succ^gaz$, and thus $by\succ^gaz$ for some $b\in ]0,1[$. Because $y=y(t(x,y);x,y)$, we have that $x\succ by$. Because $y(t(by,az);by,az)\gg az$, $by\succ az$. Because $y(t(az,x);az,x)\gg x$, $az\succ x$, which contradicts the transitivity of $\succsim$. Therefore, $\succsim^g$ is transitive. By Lemma \ref{LEM5}, $g$ satisfies Ville's axiom. This completes the proof of III).

Fifth, suppose that $g$ satisfies the weak axiom. Suppose that $x,y\in f^g(p,m)$ and $x\neq y$. Then, $g(x)=\lambda p$, $g(y)=\mu p$ for some $\lambda,\mu>0$, and $p\cdot x=p\cdot y=m$. This implies that
\[g(x)\cdot y=\lambda p\cdot y=\lambda m=\lambda p\cdot x=g(x)\cdot x,\]
\[g(y)\cdot x=\mu p\cdot x=\mu m=\mu p\cdot y=g(y)\cdot y,\]
which contradicts the weak axiom. Therefore, $f^g(p,m)$ is either the empty set or a singleton. Suppose that $x=f^g(p,m),\ y=f^g(q,w),\ x\neq y$, and $p\cdot y\le m$. Because $g$ satisfies the weak axiom, $f^g=f^{\succsim^g}$, and thus $x\succ^gy$. This implies that $q\cdot x>w$, and thus $f^g$ satisfies the weak axiom of revealed preference.

Conversely, suppose that $f^g$ is single-valued and satisfies the weak axiom of revealed preference. Suppose that $g(x)\cdot y\le g(x)\cdot x$ and $x\neq y$. Because $x=f^g(g(x),g(x)\cdot x)$ and $y=f^g(g(y),g(y)\cdot y)$, by the weak axiom of revealed preference, $g(y)\cdot x>g(y)\cdot y$. Therefore, $g$ satisfies the weak axiom.

Last, suppose that $g$ satisfies Ville's axiom. If $g$ satisfies the weak axiom, then $f^g=f^{u^g_v}$. Suppose that $x,y\in \Omega$, $u^g_v(y)\ge u^g_v(x)$, $y\neq x$, and $0<t<1$. Define $z=(1-t)x+ty$. Then, either $g(z)\cdot x\le g(z)\cdot z$ or $g(z)\cdot y\le g(z)\cdot z$. If the former holds, then $u^g_v(z)>u^g_v(x)$. If the latter holds, then $u^g_v(z)>u^g_v(y)\ge u^g_v(x)$. Hence, in both cases, $u^g_v(z)>u^g_v(x)$, and thus $u^g_v$ is strictly quasi-concave. Conversely, suppose that $u^g_v$ is strictly quasi-concave. By III), $f^g=f^{u^g_v}$, and thus $f^g$ is single-valued and satisfies the weak axiom of revealed preference. This implies that $g$ satisfies the weak axiom, which completes the proof. $\blacksquare$

\subsection{Proof of Theorem \ref{THM2}}
Suppose that $g$ satisfies the weak weak axiom and Ville's axiom. Then, $f^g=f^{u^g_v}$. Let $(p,m)\in (\mathbb{R}^n_+\setminus \{0\})\times \mathbb{R}_{++}$, and suppose that $f^g(p,m)=x^*$. Choose $x_0\in \Delta(p,m)$ and let $x(t)$ be a solution to (\ref{IMP}). By i) of Lemma \ref{LEM4}, we have that $u^g_v(x(t))$ is increasing, and thus, $L(x)=u^g_v(x^*)-u^g_v(x(t))$ is a Lyapunov function. Therefore, any improvement process is both locally and compact stable.

Next, suppose that $g$ satisfies the Ville's axiom but violates the weak weak axiom. By Lemma \ref{LEM3}, $\succsim^g$ is not weakly convex, and thus, $u^g_v$ is not quasi-concave. Therefore, there exist $x,y\in \Omega$ and $t\in ]0,1[$ such that $y\succsim^gx$ and $x\succ^g(1-t)x+ty$. Choose $t^*=\max \arg \min\{u^g_v((1-t)x+ty)|t\in [0,1]\}$, and let $z=(1-t^*)x+t^*y$. Define $p=g(z)$ and $m=g(z)\cdot z$. If $g(z)\cdot (y-x)\neq 0$, then $t\mapsto u^g_v((1-t)x+ty)$ is either increasing or decreasing around $t^*$, which contradicts the definition of $t^*$. Therefore, we have that $g(z)\cdot (y-x)=0$, and thus, $p\cdot y=p\cdot x=m$. Hence, for any open neighborhood $U$ of $z$, there exists $x_0\in \Delta(p,m)$ such that $u^g_v(x_0)>u^g_v(z)$. Let $x(t)$ be a solution to (\ref{IMP}). By i) of Lemma \ref{LEM4}, $u^g_v(x(t))$ is increasing, and thus this process does not satisfies local stability.

Third, suppose that $g$ violates the Ville's axiom. In the Step 1 of the proof of Lemma \ref{LEM5}, we constructed a Ville curve $x:[0,T]\to \Omega$ such that $x(t)=x(s)$ if and only if either $t=s$ or $t,s\in \{0,T\}$, and there exists a finite increasing sequence $t_0',t_1',t_2',t_3',t_4'$ such that $t_0'=0$, $t_4'=T$, the restriction of $x(t)$ into $[t_i',t_{i+1}']$ is $C^1$, $g(x(0))\cdot D_+x(0)>0$, $g(x(T))\cdot D_-x(T)>0$, and for each $i\in \{1,2,3\}$, $g(x(t_i'))\cdot D_-x(t_i')>0$ and $g(x(t_i'))\cdot D_+x(t_i')>0$.\footnote{Here, $D_-x(t)$ (resp. $D_+x(t)$) denotes the left-derivative (resp. the right-derivative) of $x$ at $t$.} Choose $(p,m)\in \mathbb{R}^n_{++}\times \mathbb{R}_{++}$ such that $p\cdot x(t)<m$ for all $t\in [0,T]$. Choose $\varepsilon>0$ such that $t_i'-t_{i-1}'\ge 6\varepsilon$ for all $i\in \{1,2,3,4\}$. Define
\[y(t)=\begin{cases}
x(t+3\varepsilon) & \mbox{if }0\le t\le T-3\varepsilon,\\
x(t+3\varepsilon-T) & \mbox{if }T-3\varepsilon\le t\le T,
\end{cases}\]
and $t_0=0$, $t_i=t_i'-3\varepsilon$ for $i\in \{1,2,3\}$, $t_4=T-3\varepsilon$ and $t_5=T$. Then, $y(t)$ is a Ville curve such that $D_+y(0)=D_-y(T)$, $g(y(0))\cdot D_+y(0)>0$, $g(y(T))\cdot D_-y(T)>0$, and for $i\in \{1,2,3,4\}$, $g(y(t_i))\cdot D_-y(t_i)>0$ and $g(y(t_i))\cdot D_+y(t_i)>0$. Now, choose a continuous and nonnegative function $\phi:\mathbb{R}\to \mathbb{R}_+$ such that $\phi(t)=0$ if $t\notin [\varepsilon,2\varepsilon]$, and $\int_{\varepsilon}^{2\varepsilon}\phi(t)dt=1$. Let $M=\max\{\phi(t)|\varepsilon\le t\le 2\varepsilon\}$. Because the set $\{y(t)|t_i\le t\le t_i+2\varepsilon\}$ is compact, there exists $\varepsilon'>0$ such that if $\|v\|,\|w\|\le \varepsilon'$, then $g(y(t)+v)\cdot (D_+y(t)+w)>0$ for all $t\in [t_i,t_i+2\varepsilon]$. Now, choose $\delta_i>0$ so small, and for $t\in [t_i-\delta_i,t_i+\delta_i]$, define
\[\psi_i(t)=\frac{t-(t_i-\delta_i)}{2\delta_i}\dot{y}(t_i+\delta_i)+\frac{t_i+\delta_i-t}{2\delta_i}\dot{y}(t_i-\delta_i),\]
\[z_i(t)=y(t_i-\delta_i)+\int_{t_i-\delta_i}^t\psi_i(s)ds.\]
If $\delta_i$ is sufficiently small, then $\delta_i<\varepsilon$, $g(z_i(t))\cdot \psi_i(t)>0$ for all $t\in [t_i-\delta_i,t_i+\delta_i]$, and for $v=z_i(t_i+\delta_i)-y(t_i+\delta_i)$, $\|v\|\le \frac{\varepsilon'}{M+1}$. Define $z_i(t)=y(t)+v$ if $t\in [t_i+\delta_i,t_i+\varepsilon]$, and
\[z_i(t)=y(t)+v-\int_{t_i+\varepsilon}^t\phi(s-t_i)vds,\]
for all $t\in [t_i+\varepsilon,t_i+2\varepsilon]$. Then, $z_i(t_i+2\varepsilon)=y(t_i+2\varepsilon)$, $D_-z_i(t_i+2\varepsilon)=\dot{y}(t_i+2\varepsilon)$, and for every $t\in [t_i-\delta_i,t_i+2\varepsilon]$, $g(z_i(t))\cdot \dot{z}_i(t)>0$. Define
\[z(t)=\begin{cases}
z_i(t) & \mbox{if }t_i-\delta_i\le t\le t_i+2\varepsilon\mbox{ for some }i\in \{1,2,3,4\},\\
y(t) & \mbox{otherwise}.
\end{cases}\]
Then, $z(t)$ is a $C^1$ Ville curve such that $D_+z(0)=D_-z(T)$.

Let $X=\{z(t)|0\le t\le T\}$. Then, $X$ is a $1$-dimensional compact $C^1$ manifold. Define $d(x)=\inf\{\|x-y\||y\in X\}$. By Fact \ref{FACT5}, there exist $\varepsilon_1>0$ and a $C^1$ function $\pi:U\to X$ such that $\pi(x)=x$ for all $x\in X$, where $U=\{x\in \Omega|d(x)\le 2\varepsilon_1\}\subset \Omega$. Define $h_1(x)=\dot{z}(t)$ if $\pi(x)=z(t)$. If $\varepsilon_2>0$ is sufficiently small, then $g(x)\cdot h_1(x)>0$ for all $x\in U$ such that $d(x)\le 2\varepsilon_2$. Also, define
\[h_2(x)=g(x)-\frac{p\cdot x}{m}\frac{p\cdot g(x)}{\|p\|^2}p.\]
Then, $g(x)\cdot h_2(x)>0$ whenever $x\in \Delta(p,m)$ and $x\neq f^g(p,m)$. Let $\psi:\mathbb{R}\to [0,1]$ be a continuous function such that $\psi(t)=0$ if $t\le 0$ and $\psi(t)=1$ if $t\ge \varepsilon_2$, and define
\[h(x)=(1-\psi(d(x)))h_1(x)+\psi(d(x))h_2(x).\]
Then, $h(x)$ is a continuous function such that $g(x)\cdot h(x)>0$ whenever $x\in \Delta(p,m)\setminus \{f^g(p,m)\}$, and $z(t)$ is a solution to (\ref{IMP}). This implies that compact stability is violated. This completes the proof. $\blacksquare$

\subsection{Proof of Theorem \ref{THM3}}
Suppose that $g$ satisfies the weak weak axiom. Choose any $x\in \Omega$, and let $v\neq 0$ and $v\cdot g(x)=0$. Define $x(t)=x+tv$. Then,
\begin{align*}
&~\liminf_{t\downarrow 0}\frac{1}{t}v\cdot (g(x(t))-g(x))\\
=&~\liminf_{t\downarrow 0}\frac{1}{t^2}(x(t)-x)\cdot (g(x(t))-g(x))\le 0,
\end{align*}
which implies that condition (A1) holds at $x$.

Conversely, suppose that (A1) holds at every $x\in \Omega$. By Lemma \ref{LEM3}, if $\succsim^g$ is weakly convex, then $g$ satisfies the weak weak axiom. Hence, it suffices to show that $\succsim^g$ is weakly convex. Suppose not. Then, there exist $x,y\in \Omega$ and $t\in [0,1]$ such that $y\succsim^gx$ and $x\succ^g(1-t)x+ty$. It is obvious that $x$ is not proportional to $y$. Let $V^*=\mbox{span}\{x,y\}$ and $P=P(x,y)$ is orthogonal projection from $\mathbb{R}^n$ onto $V^*$, as defined in Lemma \ref{LEM1}. By Lemma \ref{LEM2}, for $z,v\in V^*\cap \Omega$, $z\succsim^gv$ if and only if $u^g(z,y)\ge u^g(v,y)$. Let $x(t)=(1-t)x+ty$, $t^*=\max\arg\min\{u^g(x(t),y)|t\in [0,1]\}$, and $z^*=x(t^*)$. Then, $0<t^*<1$, $u^g(x(t),y)\ge u^g(z^*,y)$ if $t\in [0,1]$, and $u^g(x(t),y)>u^g(z^*,y)$ if $t^*<t\le 1$. Define $p^*=Pg(z^*)$, $\Phi(s,a)=y(s;z^*,y)+ap^*$, and $\Psi(b,c)=z^*+b(y-x)+cp^*$. Then, $\Psi$ is a $C^{\infty}$ diffeomorphism between $\mathbb{R}^2$ and $V$. Moreover, $D(\Psi^{-1}\circ \Phi)(0,0)$ is regular. Therefore, by the inverse function theorem, there exist open neighborhoods $U,V$ of $(0,0)$ such that $\Psi^{-1}\circ \Phi:U\to V$ is a $C^1$ diffeomorphism. Hence, $\Phi=\Psi\circ (\Psi^{-1}\circ \Phi)$ is a $C^1$ bijection on $U$ such that $\Phi^{-1}=\Psi^{-1}\circ (\Psi^{-1}\circ \Phi)^{-1}$ is also $C^1$ on $W\equiv \Phi(U)$. Because $g$ is locally Lipschitz, we can assume without loss of generality that there exists $L>0$ such that if $z,v\in W$, then $\|g(z)-g(v)\|\le L\|z-v\|$. Moreover, because $p^*\cdot g(z^*)=\|p^*\|^2>0$, we can assume without loss of generality that $p^*\cdot g(z)>0$ for all $z\in W$. Let $(w(t),a(t))=\Phi^{-1}(x(t))$. Because $W,U$ are open, there exists $\delta>0$ such that if $|t-t^*|<\delta$, then $(w(t),0), (w(t),a(t))\in U$. Let
\[y(t)=\Phi(w(t),0)=x(t)-a(t)p^*,\ z(t)=y'(t)=(y-x)-a'(t)p^*.\]
By definition, $x(t)=y(w(t);z^*,y)+a(t)p^*$, and thus $y(t)=y(w(t);z^*,y)$. Because $y(s;z^*,y)\sim^gz^*$ for all $s$, we have that $u^g(y(t),y)=u^g(z^*,y)$, and thus $a(t)\ge 0$ for all $t\in [0,1]$ such that $|t-t^*|<\delta$. Because $a(t^*)=0$ by definition, $a(t^*)$ attains the local minimum, and thus $a'(t^*)=0$. Moreover, by definition of $y(s;z^*,y)$, we have that $\dot{y}(w(t);z^*,y)\cdot g(y(t))=0$, and thus $z(t)\cdot g(y(t))=0$. Now, choose $t\in [0,1]$ such that $0<t-t^*<\delta$, and define $\hat{z}(t,t')=y(t)+z(t)(t'-t)$. Then,
\begin{align*}
&~\frac{|(z(t')-z(t))\cdot (g(y(t'))-g(y(t)))|}{t'-t}\\
\le&~|a'(t')-a'(t)|\|p^*\|L\frac{\|y(t')-y(t)\|}{t'-t}\to 0\mbox{ as }t'\to t,
\end{align*}
and
\begin{align*}
&~\frac{z(t)\cdot (g(y(t'))-g(\hat{z}(t,t')))}{t'-t}\\
=&~L\|z(t)\|\|p^*\|\frac{|a(t')-a(t)-a'(t)(t'-t)|}{t'-t}\to 0\mbox{ as }t'\to t.
\end{align*}
Therefore,
\begin{align*}
&~(p^*\cdot g(y(t)))\limsup_{t'\downarrow t}\frac{a'(t)-a'(t')}{t'-t}\\
=&~\limsup_{t'\downarrow t}\frac{(z(t')-z(t))\cdot g(y(t))}{t'-t}=\limsup_{t'\downarrow t}\frac{z(t')\cdot g(y(t))}{t'-t}\\
=&~-\liminf_{t'\downarrow t}\frac{z(t')\cdot (g(y(t'))-g(y(t)))}{t'-t}\\
=&~-\liminf_{t'\downarrow t}\frac{1}{t'-t}[(z(t')-z(t))\cdot (g(y(t'))-g(y(t)))\\
&~+z(t)\cdot (g(\hat{z}(t,t'))-g(y(t)))+z(t)\cdot (g(y(t'))-g(\hat{z}(t,t')))]\ge 0,
\end{align*}
by condition (A1). Therefore, we have that
\[\limsup_{t'\downarrow t}\frac{a'(t)-a'(t')}{t'-t}\ge 0.\]
By the symmetrical arguments, we can show that
\[\limsup_{t'\uparrow t}\frac{a'(t)-a'(t')}{t'-t}\ge 0.\]
Now, choose $t_1\in [0,1]$ such that $0<t_1-t^*<\delta$, and define $h(t)=a'(t)(t_1-t^*)-a'(t_1)(t-t^*)$. Then, $h(t^*)=h(t_1)=0$, and thus, there exists $t^+\in ]t^*,t_1[$ such that $h(t^+)$ attains either the maximum or the minimum on $[t^*,t_1]$. If $h(t^+)$ attains the minimum, then
\[0\ge \limsup_{t\downarrow t^+}\frac{h(t^+)-h(t)}{t-t^+}=(t_1-t^*)\limsup_{t\downarrow t^+}\frac{a'(t^+)-a'(t)}{t-t^+}+a'(t_1)\ge a'(t_1),\]
which implies that $a'(t_1)\le 0$. If $h(t^+)$ attains the maximum, then by the symmetrical argument, we can also show that $a'(t_1)\le 0$. Therefore, $a'(t)\le 0$ if $1\ge t>t^*$ and $|t-t^*|<\delta$. This implies that if $t\in [0,1]$ and $|t-t^*|<\delta$, then $a(t)\le a(t^*)=0$, and thus, $a(t)\equiv 0$. Hence, $u^g(x(t),y)=u^g(z^*,y)$ for all $t>t^*$ such that $t-t^*$ is sufficiently small, which is a contradiction. Thus, $\succsim^g$ is weakly convex, as desired. In conclusion, we obtain that $g$ satisfies the weak weak axiom if and only if condition (A1) holds everywhere, and i) holds.

Next, suppose that condition (A2) holds everywhere. Then, condition (A1) also holds everywhere, and thus $g$ satisfies the weak weak axiom. Suppose that $g$ violates the weak axiom. Then, there exists $x,y\in\Omega$ such that $g(x)\cdot y\le g(x)\cdot x$ and $g(y)\cdot x\le g(y)\cdot y$. If either inequality is strong, then $g$ violates the weak weak axiom, which is a contradiction. Therefore, we have that $g(x)\cdot y=g(x)\cdot x$ and $g(y)\cdot x=g(y)\cdot y$. Let $v=y-x$ and $x(t)=x+tv$. Choose any $t\in [0,1]$. Because $g(x)\cdot x(t)=g(x)\cdot x$, by Lemma \ref{LEM3}, $x\succsim^gx(t)$. Because $g(y)\cdot x=g(y)\cdot y$, again by Lemma \ref{LEM3}, $y\succsim^gx$. Because $\succsim^g$ is weakly convex, $x(t)\succsim^gx$. Therefore, $x\sim^gx(t)$. By the same argument, we have that $y\sim^gx(t)$. If $g(x(t))\cdot x<g(x(t))\cdot x(t)$, then $x(t)\succsim^gax$ for some $a>1$, and by monotonicity and p-transitivity, $x(t)\succ^gx$, which is a contradiction. If $g(x(t))\cdot x>g(x(t))\cdot x(t)$, then $g(x(t))\cdot y<g(x(t))\cdot x(t)$, which leads a contradiction by the same logic as above. Therefore, $g(x(t))\cdot v=0$ for all $t\in [0,1]$. This implies that
\[\liminf_{t\downarrow 0}\frac{1}{t}v\cdot (g(x(t))-g(x))=0,\]
which contradicts condition (A2). Hence, ii) holds.

Now, suppose that $g$ satisfies Ville's axiom. By Lemma \ref{LEM5}, for any $v\in \Omega$, $u^g_v$ satisfies property (F). Suppose that $u^g_v$ is differentiable at $x\in \Omega$. Let $v(s)=(u^g(x,v)+s)v$, and $a(s)=u^g(v(s),x)$. Then, $a(0)=1$ and $u^g_v(x)+s=u^g_v(v(s))=u^g_v(a(s)x)$. Note that, $a(s)$ is locally Lipschitz, and thus, there exists $\varepsilon>0$ and $L>0$ such that if $s_1,s_2\in [0,\varepsilon]$, then $|a(s_1)-a(s_2)|\le L|s_1-s_2|$. Because
\[1=\limsup_{s\to 0}\frac{u^g_v(a(s)x)-u^g_v(x)}{s}\le LDu^g_v(x)x,\]
we have that $Du^g_v(x)\neq 0$. This implies that there exists $\lambda(x)>0$ such that $Du^g_v(x)=\lambda(x)g(x)$, and thus $(u^g_v,\lambda)$ is a normal solution to (\ref{TDE}) on $\Omega$. By Fact \ref{FACT6}, $g$ satisfies condition (B) almost everywhere.

Conversely, suppose that $g$ satisfies condition (B). First, choose any normal solution $(u,\lambda)$ to (\ref{TDE}) defined on some open and convex set $U_0\subset \Omega$. We show that $u$ is increasing. Suppose that $x,y\in U_0$ and $y\gg x$. Let $w=y-x$. Because $(u,\lambda)$ is a normal solution, by Fubini's theorem, there exists a sequence $(x^k)$ such that $x^k\to x$ as $k\to \infty$ and $\nabla u(x^k+tw)=\lambda(x^k+tw)g(x^k+tw)$ for almost all $t\in [0,1]$. Therefore, $u(x^k+w)>u(x^k)$, and thus, $u(y)\ge u(x)$. This implies that $u$ is nondecreasing. Suppose that $u(y)=u(x)$. Then, $u(x+tw)$ is constant on $[0,1]$, and thus, if $X=u^{-1}(u(x))$, then $w\in T_x(X)$. However, $w\cdot g(x)>0$, which is a contradiction. Therefore, $u(y)>u(x)$, and $u$ is increasing. In particular, $u$ satisfies property (F).

Second, we show that if $t^*>0$ and $y([0,t^*];x,w)\subset U_0$, then $u(x)=u(y(t;x,w))$ for all $t\in [0,t^*]$. Consider the following differential equation:
\[\dot{z}(t)=(g(z(t))\cdot x)w-(g(z(t))\cdot w)x+hw,\ z(0)=x.\]
Let $z(t;h)$ be the solution function of this equation. Note that, $z(t;0)=y(t;x,w)$, and thus there exists $\varepsilon>0$ such that if $|h|<\varepsilon$, then $z(\cdot;h)$ is defined on $[0,t^*]$ and $z(t;h)\in U_0$ for all $t\in [0,t^*]$. By i) of Lemma \ref{LEM4}, we have that if $h>0$, then $u(z(t;h))$ is increasing in $t$, and thus $u(z(t;h))\ge u(x)$ for all $t\in [0,t^*]$. Letting $h\downarrow 0$, we have that $u(z(t;0))\ge u(x)$. By the symmetrical argument, we can show that if $h<0$, then $u(z(t;h))$ is decreasing in $t$, and letting $h\uparrow 0$, $u(z(t;0))\le u(x)$. Therefore, our claim holds. In particular, if we can choose $t^*=t(x,w)$, then $u(x)=u(u^g(x,w)w)$.

Third, choose any $x\in \Omega$. If $x$ is not proportional to $v$, then by Lemma \ref{LEM2}, $u^g_v$ is locally Lipschitz around $x$. Suppose that $x=av$ for some $a>0$. Because $g$ satisfies condition (B), there exists a normal solution $(u,\lambda)$ to (\ref{TDE}) that is defined on an open and convex neighborhood $U_1$ of $x$. Choose an open and convex neighborhood $U_2\subset U_1$ of $x$ such that if $y,z\in U_2$, then $y_1(y,z),y_2(y,z)\in U_1$. By v) of Lemma \ref{LEM1}, $y([0,t(y,z)];y,z)\subset \Delta(y,z)\subset U_1$. Again, choose an open and convex neighborhood $U_3\subset U_2$ of $x$ such that if $y,z\in U_3$, then $y_1(y,z),y_2(y,z)\in U_2$. Suppose that $y,z,w\in U_3$. Then, we can easily confirm that
\[u(u^g(y,w)w)=u(y)=u(u^g(y,z)z)=u(u^g(u^g(y,z)z,w)w),\]
and because $u$ is increasing, we obtain that
\[u^g(y,w)=u^g(u^g(y,z)z,w).\]
Choose any $y\in U_3$ that is not proportional to $v$. Note that, $u^g(z,x)=\frac{\|v\|}{\|x\|}u^g_v(z)$. Therefore, for any $z\in U_3$,
\[u^g_v(z)=u^g(u^g(z,y)y,v),\]
where the right-hand side is locally Lipschitz in $z$ around $x$. Therefore, $u^g_v$ is locally Lipschitz on $\Omega$. Moreover, for $y,z\in U_3$,
\begin{align*}
u(y)\ge u(z)\Leftrightarrow&~u(u^g(y,x)x)\ge u(u^g(z,x)x)\\
\Leftrightarrow&~u^g(y,x)\ge u^g(z,x)\Leftrightarrow u^g_v(y)\ge u^g_v(z).
\end{align*}
which implies that $u^{-1}(u(y))=(u^g_v)^{-1}(u^g_v(y))$ for any $y\in U_3$, and thus, for any $y\in U_3$ and $Y(y)\equiv (u^g_v)^{-1}(u^g_v(y))\cap U_3$, $Y(y)$ is an $n-1$ dimensional $C^1$ manifold, and $g(y)$ is orthogonal to $T_y(Y(y))$.

Let $X$ be the set of all $x\in \Omega$ such that there exists an open and convex neighborhood $U_0$ of $x$ such that for any $y\in U_0$, $Y(y)\equiv (u^g_v)^{-1}(u^g_v(y))\cap U_0$ is an $n-1$ dimensional $C^1$ manifold, and $g(y)$ is orthogonal to $T_y(Y(y))$. Suppose that $x\in X$, and $U_0$ is a corresponding open and convex neighborhood of $x$. We show that if $u^g_v$ is differentiable at $y\in U_0$, then there exists $\lambda(y)>0$ such that $\nabla u^g_v(y)=\lambda(y)g(y)$. First, both $\nabla u^g_v(y), g(y)$ are orthogonal to $T_y(Y(y))$, and thus there exists $\lambda(y)\in \mathbb{R}$ such that $\nabla u^g_v(y)=\lambda(y)g(y)$. Let $a>1$. If $y$ is proportional to $v$, then $u^g_v(ay)=au^g_v(y)$, which implies that $\nabla u^g_v(y)\cdot y=1$ and thus $\lambda(y)>0$. Suppose that $y$ is not proportional to $v$. By Lemma \ref{LEM2},
\[a=u^g(ay,y)=u^g(u^g_v(ay)v,y),\]
and $u^g$ is locally Lipschitz around $(u^g_v(y)v,y)$. Therefore, there exists $\varepsilon>0$ and $L>0$ such that if $a-1<\varepsilon$, then\footnote{Note that, by Lemma \ref{LEM2}, both side of this inequality are positive.}
\[a-1=u^g(u^g_v(ay)v,y)-u^g(u^g_v(y)v,y)\le L\|v\|(u^g_v(ay)-u^g_v(y)).\]
Dividing both sides by $a-1$ and letting $a\downarrow 1$, we have that
\[1\le L\|v\|Du^g_v(y)y,\]
and thus $\nabla u^g_v(y)y>0$, which implies that $\lambda(y)>0$. Because of Rademacher's theorem, $u^g_v$ is differentiable almost everywhere, and thus $(u^g_v,\lambda)$ is a normal solution to (\ref{TDE}) on $U_0$. In particular, $u^g_v$ is increasing on $U_0$, and if $y([0,t(y,w)];y,w)\subset U_0$, then $u^g_v(u^g(y,w)w)=u^g_v(y)$.

Now, choose any $x\in \Omega$. Define $A$ as the trajectory of $y(\cdot;x,v)$ and $B=A\cap X$. We have already shown that $v^*=u^g(x,v)v\in X$. If $x$ is proportional to $v$, then $v^*=x$, and thus $x\in X$. Hence, hereafter, we assume that $x$ is not proportional to $v$. By definition, $B$ is open in $A$. We show that $B$ is closed in $A$. Suppose that $(z^k)$ is a sequence in $B$ that converges to $z\in A$ as $k\to \infty$. If $z=v^*$, then $z\in B$, and thus we assume that $z\neq v^*$. Because $g$ satisfies condition (B), there exists a normal solution $(u,\lambda)$ defined on some open and convex neighborhood $U_1$ of $z$. Note that, $A$ is the trajectory of $y(\cdot;z,v)$, and thus there exists $t^k\in \mathbb{R}$ such that $z^k=y(t^k;z,v)$. Choose an open and convex neighborhood $U_2\subset U_1$ of $z$ such that if $y,w\in U_2$, then $y_1(y,w),y_2(y,w)\in U_2$, and thus $y([0,t(y,w)];y,w)\subset U_1$. Then, $u(y)=u(u^g(y,w)w)$ for any $y,w\in U_2$. Because $(z^k)$ converges to $z$, there exists $k$ such that $z^k\in U_2$. Because $z^k\in B$, there exist an open and convex neighborhood $U_3\subset U_2$ of $z^k$ and $\mu:U_3\to \mathbb{R}_{++}$ such that $(u^g_v,\mu)$ is a normal solution to (\ref{TDE}) on $U_3$. Choose an open and convex neighborhood $U_4\subset U_3$ of $z^k$ such that if $y,w\in U_4$, then $y_1(y,w),y_2(y,w)\in U_3$, and thus $y([0,t(y,w)];y,w)\subset U_3$. Then, for any $y,w\in U_4$, $u(y)=u(u^g(y,w)w)$ and $u^g_v(y)=u^g_v(u^g(y,w)w)$. Because $y(t;y,v)$ is continuous in $y$, there exists an open and convex neighborhood $U_5\subset U_2$ of $z$ such that if $y\in U_5$, then $y(t^k;y,v)\in U_4$. Choose an open and convex neighborhood $U_6\subset U_5$ of $z$ such that if $y,w\in U_6$, then $y_1(y,w),y_2(y,w)\in U_5$, and thus $y([0,t(y,w)];y,w)\subset U_5$. Suppose that $y,w\in U_6$ and let $w^*=u^g(y,w)w$. Then, $y,w^*\in U_5$, and $u(y)=u(w^*)$. Let $y'=y(t^k;y,v)$ and $w'=y(t^k;w^*,v)$. Because of Lemma \ref{LEM2}, $u^g_v(y')=u^g_v(y)$ and $u^g_v(w')=u^g_v(w^*)$. Because $y\in U_5$, we have that $y'\in U_4$, and thus $y,y'\in U_2$. Therefore, $u(y)=u(y')$. By the same reason, $u(w^*)=u(w')$, and thus $u(w')=u(y')$. Because $y',w'\in U_4$, $y([0,t(y',w')];y',w')\subset U_3$, and thus $u^g_v(y')=u^g_v(u^g(y',w')w')$ and $u(y')=u(u^g(y',w')w')$. Because $u$ is increasing, $u^g(y',w')=1$, and thus $u^g_v(y')=u^g_v(w')$. In conclusion, we obtain that $u^g_v(y)=u^g_v(u^g(y,w)w)$ and $u(y)=u(u^g(y,w)w)$ for any $y,w\in U_6$. Because both $u(aw)$ and $u^g_v(aw)$ are increasing in $a$, we have that $(u^g_v)^{-1}(u^g_v(y))\cap U_6=u^{-1}(u(y))\cap U_6$. This implies that $z\in B$, as desired.

Hence, $B$ is a nonempty, open, and closed subset of $A$. Because $A$ is path-connected, we have that $B=A$. This implies that $x\in X$. Because $x$ is arbitrary, we have that $X=\Omega$. We have already shown that there exists $\lambda:\Omega\to \mathbb{R}_{++}$ such that if $u^g_v$ is differentiable at $x$, then $\nabla u^g_v(x)=\lambda(x)g(x)$. Therefore, $u^g_v$ is a normal solution to (\ref{TDE}) defined on $\Omega$, which implies that $u^g_v$ is increasing. To summarize the above argument, we have shown that $u^g_v$ satisfies property (F). By Lemma \ref{LEM5}, we conclude that $g$ satisfies Ville's axiom. Hence, iii) holds. This completes the proof. $\blacksquare$

\subsection{Proof of Proposition \ref{PROP1}}
First, we prove the following lemma.

\begin{lem}\label{LEM6}
Suppose that $g:\Omega\to \mathbb{R}^n_+\setminus \{0\}$ and $c:\Omega\to \mathbb{R}_{++}$ are locally Lipschitz functions, and let $h(x)=c(x)g(x)$. If both $g$ and $c$ are differentiable at $x$, then $g$ satisfies condition (A1) at $x$ if and only if $h$ satisfies condition (A1), $g$ satisfies condition (A2) at $x$ if and only if $h$ satisfies condition (A2), and $g$ satisfies condition (B) at $x$ if and only if $h$ satisfies condition (B).
\end{lem}

\begin{proof}
First, note that $v\cdot g(x)=0$ if and only if $v\cdot h(x)=0$. If $v\cdot g(x)=0$,
\[\liminf_{t\downarrow 0}\frac{1}{t}v\cdot (g(x+tv)-g(x))=v^TDg(x)v,\]
and,
\[\liminf_{t\downarrow 0}\frac{1}{t}v\cdot (h(x+tv)-h(x))=v^TDh(x)v.\]
Now,
\[Dh(x)=g(x)Dc(x)+c(x)Dg(x),\]
and thus,
\[v^TDh(x)v=c(x)v^TDg(x)v.\]
Therefore, $g$ satisfies condition (A1) if and only if $h$ satisfies condition (A1), and $g$ satisfies condition (A2) if and only if $h$ satisfies condition (A2). Moreover,
\begin{align*}
&~h_i\left(\frac{\partial h_j}{\partial x_k}-\frac{\partial h_k}{\partial x_j}\right)+h_j\left(\frac{\partial h_k}{\partial x_i}-\frac{\partial h_i}{\partial x_k}\right)+h_k\left(\frac{\partial h_i}{\partial x_j}-\frac{\partial h_j}{\partial x_i}\right)\\
=&~c^2g_i\left(\frac{\partial g_j}{\partial x_k}-\frac{\partial g_k}{\partial x_j}\right)+g_j\left(\frac{\partial g_k}{\partial x_i}-\frac{\partial g_i}{\partial x_k}\right)+g_k\left(\frac{\partial g_i}{\partial x_j}-\frac{\partial g_j}{\partial x_i}\right)\\
&~+c\left[(g_ig_j-g_jg_i)\frac{\partial c}{\partial x_k}+(g_jg_k-g_kg_j)\frac{\partial c}{\partial x_i}+(g_kg_i-g_ig_k)\frac{\partial c}{\partial x_j}\right]\\
=&~c^2\left[g_i\left(\frac{\partial g_j}{\partial x_k}-\frac{\partial g_k}{\partial x_j}\right)+g_j\left(\frac{\partial g_k}{\partial x_i}-\frac{\partial g_i}{\partial x_k}\right)+g_k\left(\frac{\partial g_i}{\partial x_j}-\frac{\partial g_j}{\partial x_i}\right)\right].
\end{align*}
Therefore, $g$ satisfies condition (B) if and only if $h$ satisfies condition (B). This completes the proof. $\blacksquare$
\end{proof}

Now, let $c(x)=\frac{1}{g_n(x)}$. Then, $h(x)=\bar{g}(x)$. Therefore, we can assume that $g\equiv \bar{g}$. For $\tilde{g}=(g_1,...,g_{n-1})$,
\[Dg(x)=\left(
\begin{array}{c|c}
A_g(x) & \frac{\partial \tilde{g}}{\partial x_n}(x)\\ \hline
0^T & 0
\end{array}
\right)
+
\left(
\begin{array}{c|c}
\frac{\partial \tilde{g}}{\partial x_n}(x)\tilde{g}^T(x) & 0\\ \hline
0^T & 0
\end{array}
\right).\]
Now, suppose that $v=(\tilde{v},v_n)\in \mathbb{R}^n$ and $v\cdot g(x)=0$. Then, $\tilde{v}\neq 0$ if and only if $v\neq 0$. Moreover,
\begin{align*}
v^TDg(x)v=&~\tilde{v}^TA_g(x)\tilde{v}+\tilde{v}^T\frac{\partial \tilde{g}}{\partial x_n}v_n+\tilde{v}^T\frac{\partial \tilde{g}}{\partial x_n}(x)\tilde{g}^T(x)\tilde{v}\\
=&~\tilde{v}^TA_g(x)\tilde{v},
\end{align*}
and thus, $g$ satisfies condition (A1) at $x$ if and only if $A_g(x)$ is negative semi-definite, and $g$ satisfies condition (A2) at $x$ if and only if $A_g(x)$ is negative definite.

Now, we show that if $g_n(x)\neq 0$, then $g$ satisfies condition (B) at $x$ if and only if for each $i,j\in \{1,...,n-1\}$,
\begin{equation}\label{INT}
g_i\left(\frac{\partial g_j}{\partial x_n}-\frac{\partial g_n}{\partial x_j}\right)+g_j\left(\frac{\partial g_n}{\partial x_i}-\frac{\partial g_i}{\partial x_n}\right)+g_n\left(\frac{\partial g_i}{\partial x_j}-\frac{\partial g_j}{\partial x_i}\right)=0.
\end{equation}
If $g$ satisfies condition (B) at $x$, then clearly (\ref{INT}) holds at $x$ for all $i,j\in \{1,...,n-1\}$. Conversely, suppose that (\ref{INT}) holds for all $i,j\in \{1,...,n-1\}$ at $x$. Then, for $i,j,k\in \{1,...,n-1\}$,
\begin{align*}
g_k\left[g_i\left(\frac{\partial g_j}{\partial x_n}-\frac{\partial g_n}{\partial x_j}\right)+g_j\left(\frac{\partial g_n}{\partial x_i}-\frac{\partial g_i}{\partial x_n}\right)+g_n\left(\frac{\partial g_i}{\partial x_j}-\frac{\partial g_j}{\partial x_i}\right)\right]=&~0,\\
g_i\left[g_j\left(\frac{\partial g_k}{\partial x_n}-\frac{\partial g_n}{\partial x_k}\right)+g_k\left(\frac{\partial g_n}{\partial x_j}-\frac{\partial g_j}{\partial x_n}\right)+g_n\left(\frac{\partial g_j}{\partial x_k}-\frac{\partial g_k}{\partial x_j}\right)\right]=&~0,\\
g_j\left[g_k\left(\frac{\partial g_i}{\partial x_n}-\frac{\partial g_n}{\partial x_i}\right)+g_i\left(\frac{\partial g_n}{\partial x_k}-\frac{\partial g_k}{\partial x_n}\right)+g_n\left(\frac{\partial g_k}{\partial x_i}-\frac{\partial g_i}{\partial x_k}\right)\right]=&~0.
\end{align*}
Summing up these equations, we have that
\[g_n\left[g_i\left(\frac{\partial g_j}{\partial x_k}-\frac{\partial g_k}{\partial x_j}\right)+g_j\left(\frac{\partial g_k}{\partial x_i}-\frac{\partial g_i}{\partial x_k}\right)+g_k\left(\frac{\partial g_i}{\partial x_j}-\frac{\partial g_j}{\partial x_i}\right)\right]=0,\]
and because $g_n\neq 0$, we have that condition (B) holds at $x$. In particular, if $g_n(x)\equiv 1$, then
\[g_i\left(\frac{\partial g_j}{\partial x_n}-\frac{\partial g_n}{\partial x_j}\right)+g_j\left(\frac{\partial g_n}{\partial x_i}-\frac{\partial g_i}{\partial x_n}\right)+g_n\left(\frac{\partial g_i}{\partial x_j}-\frac{\partial g_j}{\partial x_i}\right)=a_{ij}-a_{ji},\]
and thus, (\ref{INT}) holds at $x$ if and only if $A_g(x)$ is symmetric. This completes the proof. $\blacksquare$

\subsection{Proof of Corollary \ref{COR1}}
There exists a continuously differentiable function $g^*:\Omega\to \mathbb{R}^n_{++}$ whose Antonelli matrix is not symmetric.\footnote{The existence of such a function $g^*$ is shown by Gale (1960) in the case of $n=3$, and his construction method can easily be extended to the case in which $n\ge 4$.} If $g^+(x)=\frac{1}{\sum_{i=1}^ng_i^*(x)}g^*(x)$, then $g^+\in \mathscr{G}$, and by Lemma \ref{LEM6}, the Antonelli matrix of $g^+$ is not symmetric. Therefore, if $g\in \mathscr{G}$ satisfies Ville's axiom and $U$ is a neighborhood of $g$, then $(1-t)g+tg^*$ violates Ville's axiom whenever $0<t\le 1$. Because $(1-t)g+tg^*\in U$ if $t>0$ is sufficiently small, our claim holds. This completes the proof. $\blacksquare$

\subsection{Proof of Theorem \ref{THM4}}
By Lemma \ref{LEM2}, $\succsim^g$ is a complete, p-transitive, continuous, and monotone binary relation on $\Omega$. 

Suppose that $\succsim$ is a complete, p-transitive, and continuous binary relation such that $f^g=f^{\succsim}$. First, choose any $x,v\in \Omega$ such that $x$ is not proportional to $v$. As in the proof of Step 5 of Lemma \ref{LEM3}, choose $k\in \mathbb{N}$, and define
\[h^k=\frac{t(x,v)}k,\ t_i^k=ih^k,\]
\[x_0^k=x,\ x_{i+1}^k=x_i^k+h^kCRPg(x_i^k),\]
and for any $t\in [t_i^k,t_{i+1}^k]$,
\[x^k(t)=\frac{t-t_i^k}{t_{i+1}^k-t_i^k}x_{i+1}^k+\frac{t_{i+1}^k-t}{t_{i+1}^k-t_i^k}x_i^k.\]
As we have mentioned, $x^k(t)$ uniformly converges to $y(t;x,v)$ on $[0,t(x,v)]$ as $k\to \infty$. Because $g(x_i^k)\cdot (x_{i+1}^k-x_i^k)=0$, we have that $x_i^k\succsim x_{i+1}^k$. By p-transitivity of $\succsim$, we have that $x\succsim x_k^k$. If $z=av$ for some $a<u^g(x,v)$, then there exists $k$ such that $x_k^k\gg z$. Because $g(x_k^k)\cdot z<g(x_k^k)\cdot x_k^k$, $z\notin f^{\succsim}(g(x_k^k),g(x_k^k)\cdot x_k^k)$, and thus $x_k^k\succ z$. By p-transitivity, we have that $x\succ z$. Because $\succsim$ is continuous, we have that if $z=u^g(x,v)v$, then $x\succsim z$. In particular, we have that $x\succsim^gv$, then $x\succsim v$. Next, suppose that $x\not\succsim^gv$. Then, $u^g(x,v)<1$, and thus $u^g(v,x)>1$. Applying the above argument, we have that $v\succ x$, and thus $x\not\succsim v$. Therefore, $x\succsim^gv$ if and only if $x\succsim v$, and I) is proved.

II) is already proved in Lemma \ref{LEM3} and III) is already proved in Lemma \ref{LEM5}. If $g(x)$ is not proportional to $g(x)$, then there exists $v\in \Omega$ such that
\[(g(x)\cdot x)v-(g(x)\cdot v)x\neq (h(x)\cdot x)v-(h(x)\cdot v)x,\]
which implies that there exists $z\in \Omega$ such that $x\sim^gz$ and $x\not\sim^hz$. Therefore, if $\succsim^g\neq \succsim^h$.

Conversely, suppose that $h(x)=\lambda(x)g(x)$ for any $x\in \Omega$, where $\lambda(x)>0$. Choose any $(x,v)\in \Omega^2$. If $x$ is proportional to $v$, then $u^g(x,v)=u^h(x,v)$ by definition, and thus we assume that $x$ is not proportional to $v$. Let $y(t;x,v)$ be the solution function of (\ref{IC}), and $z(t;x,v,c)$ be the solution function of the following differential equation
\[\dot{z}(t)=(h(z(t))\cdot v)x-(h(z(t))\cdot x)v+cv,\ h(0)=x.\]
Because of Lemma \ref{LEM2}, $u^h(x,v)v=z(t^h(x,v);x,v,0)$ for some $t^h(x,v)>0$. Therefore, there exists $\varepsilon>0$ such that if $|c|<\varepsilon$, then $z(\cdot;x,v,c)$ is defined on $[0,t^h(x,v)]$. Recall that $w^*(x,v)=(v\cdot x)v-(v\cdot v)x$. Then, $w^*(x,v)\cdot \dot{y}(t;x,v)>0$ and $w^*(x,v)\cdot \dot{z}(t;x,v,c)>0$, and thus, by the implicit function theorem, there exists a $C^1$ function $s(t,c)$ on $[0,t^h(x,v)]\times ]-\varepsilon,\varepsilon[$ such that $z(t;x,v,c)-y(s(t,c);x,v)$ is proportional to $v$. If $c>0$, then
\[\dot{z}(t;x,v,c)\cdot g(x)>0=g(x)\cdot \left.\frac{d}{dt}y(s(t,c);x,v)\right|_{t=0},\]
and thus, if $t>0$ is sufficiently small, then $z(t;x,v,c)\gg y(s(t,c);x,v)$. Suppose that $z(t;x,v,c)\le y(s(t,c);x,v)$ for some $t\in ]0,t^h(x,v)]$. Let $t^*=\inf\{t\in ]0,t^h(x,v)]|z(t;x,v,c)\le y(s(t,c);x,v)\}$. Then, $t^*>0$, and 
\[\dot{z}(t^*;x,v,c)\cdot g(z(t^*;x,v,c))>0=g(y(s(t^*,c);x,v))\cdot \left.\frac{d}{dt}y(s(t,c);x,v)\right|_{t=t^*}.\]
Thus, $z(t;x,v,c)\le y(s(t,c);x,v)$ for some $t<t^*$, which is a contradiction. Therefore, $z(t;x,v,c)\gg y(s(t,c);x,v)$. Letting $c\to 0$, we have that $z(t;x,v,0)\ge y(s(t,0);x,v)$. By the symmetrical argument, we can show that $z(t;x,v,0)\le y(s(t,0);x,v)$. Therefore,
\[u^h(x,v)v=z(t^h(x,v);x,v)=y(s(t^h(x,v),0);x,v)=y(t(x,v);x,v)=u^g(x,v)v,\]
which implies that $u^h=u^g$, and thus IV) holds.

It suffices to show the continuity of the mapping $g\mapsto \succsim^g$. Suppose that $(g^k)$ is a sequence in $\mathscr{G}$ that uniformly converges to $g\in \mathscr{G}$ as $k\to \infty$ on any compact set $C\subset \Omega$.

First, we show that $(u^{g^k})$ converges to $u^g$ as $k\to \infty$ on any compact set. Choose any compact set $D\subset \Omega^2$. Then, $y_1(D),y_2(D)$ are also compact, and thus, there exists a convex and compact set $C\subset \Omega$ such that $y_1(D)\cup y_2(D)$ is included in the interior of $C$ and $D$ is included in the interior of $C^2$. Let $U$ be the interior of $C$, and $\|g^*-g^+\|=\max\{\|g^*(x)-g^+(x)\||x\in C\}$. Because $g$ is locally Lipschitz, there is a constant $L>0$ such that if $x,y\in C$, then
\[\|g(x)-g(y)\|\le L\|x-y\|.\]
Consider the following differential equations:
\begin{align}
\dot{y}(t)=&~(g(y(t))\cdot x)v-(g(y(t))\cdot v)x,\ y(0)=x,\label{IC3}\\
\dot{y}(t)=&~(g^k(y(t))\cdot x)v-(g^k(y(t))\cdot v)x,\ y(0)=x.\label{IC4}
\end{align}
Let $y(t;x,v)$ be the solution function of (\ref{IC3}) and $y^k(t;x,v)$ be the solution function of (\ref{IC4}). Let $w^*(x,v)=(v\cdot x)v-(v\cdot v)x$, $t(x,v)=\min\{t\ge 0|w^*(x,v)\cdot y(t;x,v)=0\}$, and $t^k(x,v)=\min\{t\ge 0|w^*(x,v)\cdot y^k(t;x,v)\}$. We have already shown that $u^g(x,v)v=y(t(x,v);x,v)$ and $u^{g^k}(x,v)v=y^k(t^k(x,v);x,v)$. Moreover, if $(x,v)\in D$, then $x,y_1(x,v),y_2(x,v)\in U$, and thus, $y([0,t(x,v)];x,v)\subset U$ and $y^k([0,t^k(x,v)];x,v)\subset U$.

Choose any $\varepsilon>0$. Let $W=\{(x,v)\in D|\|y_1(x,v)-y_2(x,v)\|<\|v\|\varepsilon\}$. Because $u^g(x,v)v,u^{g^k}(x,v)v\in [y_1(x,v),y_2(x,v)]$, $(x,v)\in W$ implies that $|u^g(x,v)-u^{g^k}(x,v)|<\varepsilon$. Thus, define $D'=D\setminus W$. Then, $D'$ is also compact, and $(x,v)\in D'$ implies that $x$ is not proportional to $v$. By Step 9 of the proof of Lemma \ref{LEM2}, $t(x,v)$ is continuous on $D'$. Let $h(x,v)$ be the supremum of $h\in [0,1]$ such that $y(t;x,v)$ is defined on $[0,t(x,v)+h]$ and $y(t;x,v)\in U$ for all $t\in [0,t(x,v)+h]$, and $D^k=\{(x,v)\in D'|h(x,v)>\frac{1}{k}\}$. By the above consideration, we have that $D^k$ is open in $D'$, and $\cup_kD^k=D'$. Because $D'$ is compact, there exists $k$ such that $D^k=D'$. Therefore, there exists $h>0$ such that $y(\cdot;x,v)$ is defined on $[0,t(x,v)+h]$ for all $(x,v)\in D'$. If $(x,v)\in D'$, then $y([0,t(x,v)+h];x,v)\subset U$. Because the set
\[C'=\{y(t;x,v)|(x,v)\in D',\ t\in [0,t(x,v)+h]\}\]
is compact and included in $U$,
\[\delta\equiv\min\{\|y-z\||y\in C', z\notin U\}>0.\]
Choose any $c\in ]0,t(x,v)+h]$. If $y^k(\cdot;x,v)$ is defined on $[0,c]$ and $y^k([0,c];x,v)\subset U$, then for any $t\in [0,c]$,
\begin{align*}
&~\|y^k(t;x,v)-y(t;x,v)\|\\
\le&~\int_0^t\|g^k(y^k(s;x,v))-g(y(s;x,v))\|\cdot 2\|x\|\|v\|ds\\
\le&~\int_0^t\|g(y^k(s;x,v))-g(y(s;x,v))\|\cdot 2\|x\|\|v\|ds\\
&~+\int_0^t\|g^k(y^k(s;x,v))-g(y^k(s;x,v))\|\cdot 2\|x\|\|v\|ds\\
\le&~2\|x\|\|v\|\int_0^t[\|g(y^k(s;x,v))-g(y(s;x,v))\|+\|g^k-g\|]ds\\
\le&~2\|x\|\|v\|\int_0^tL\|y^k(s;x,v)-y(s;x,v)\|ds+2\|x\|\|v\|\|g^k-g\|t.
\end{align*}
By Fact \ref{FACT8},
\[\|y^k(t;x,v)-y(t;x,v)\|\le \frac{\|g^k-g\|}{L}(e^{2\|x\|\|v\|L(t(x,v)+h)}-1)\equiv \|g^k-g\|M(x,v).\]
Because $t(x,v)$ is continuous on $D'$, $M(x,v)$ is also continuous on $D'$. Let
\[M=\max\{M(x,v)|(x,v)\in D'\}.\]
If $k$ is sufficiently large, then $\|g^k-g\|M<\delta$. If $(x,v)\in D'$ and $y^k(\cdot;x,v)$ is defined only on $[0,c[$ for some $c\le t(x,v)+h$, then by Fact \ref{FACT1}, there exists $t\in [0,c[$ such that $y^k(t;x,v)\notin C$, and thus $\|y^k(t;x,v)-y(t;x,v)\|\ge \delta$, which contradicts the above inequality. Therefore, $y^k(\cdot;x,v)$ is defined on $[0,t(x,v)+h]$ and $y^k([0,t(x,v)+h];x,v)\subset C$. Let
\[\delta'=\min\{w^*(x,v)\cdot y(t(x,v)+h;x,v)|(x,v)\in D'\}>0.\]
\[w^+=\max\{\|w^*(x,v)\||(x,v)\in D'\}>0.\]
If $k$ is sufficiently large, then $\|g^k-g\|M<\frac{\delta'}{w^+}$, and thus,
\[|w^*(x,v)\cdot (y^k(t(x,v)+h;x,v)-y(t(x,v)+h;x,v))|\le\|w^*(x,v)\|\|g^k-g\|M<\delta',\]
which implies that $t^k(x,v)<t(x,v)+h$.

Now, recall that
\[w^*(x,v)\cdot \dot{y}(t;x,v)=(C(x,v))^2v\cdot g(y(t;x,v)).\]
Define
\[M_1=\min\{(C(x,v))^2v\cdot g(y)|(x,v)\in D',y\in C\}>0.\]
If $(x,v)\in D'$ and $y\in C$, then
\begin{align*}
(C(x,v))^2v\cdot g^k(y)=&~(C(x,v))^2v\cdot g(y)+(C(x,v))^2v\cdot (g^k(y)-g(y))\\
\ge&~M_1-(C(x,v))^2\|v\|\|g^k-g\|,
\end{align*}
and thus, if $k$ is sufficiently large, then
\[w^*(x,v)\cdot \dot{y}^k(t;x,v)>\frac{M_1}{2}\]
for all $(x,v)\in D'$ and all $t\in [0,t(x,v)+h]$. If $t^k(x,v)\le t(x,v)$, then
\begin{align*}
w^*(x,v)\cdot y^k(t(x,v);x,v)=&~\int_{t^k(x,v)}^{t(x,v)}[w^*(x,v)\cdot \dot{y}^k(t;x,v)]dt\\
\ge&~\frac{M_1}{2}(t(x,v)-t^k(x,v)),
\end{align*}
and if $t(x,v)\le t^k(x,v)$, then
\begin{align*}
w^*(x,v)\cdot y(t^k(x,v);x,v)=&~\int_{t(x,v)}^{t^k(x,v)}[w^*(x,v)\cdot \dot{y}(t;x,v)]dt\\
\ge&~M_1(t^k(x,v)-t(x,v)).
\end{align*}
Note that,
\begin{align*}
w^*(x,v)\cdot y^k(t(x,v);x,v)=&~w^*(x,v)\cdot (y^k(t(x,v);x,v)-y(t(x,v);x,v))\\
\le&~\|w^*(x,v)\|\|g^k-g\|M,\\
w^*(x,v)\cdot y(t^k(x,v);x,v)=&~w^*(x,v)\cdot (y(t^k(x,v);x,v)-y^k(t^k(x,v);x,v))\\
\le&~\|w^*(x,v)\|\|g^k-g\|M.
\end{align*}
Therefore, if we define
\[M_2=\max\{\|v\|^{-1}|(x,v)\in D'\},\]
\[M_3=\max\{C(x,v)\|Pg(y)\||(x,v)\in D',y\in C\},\]
\[M_4=\max\{\|w^*(x,v)\||(x,v)\in D'\},\]
then, for all $(x,v)\in D'$,
\begin{align*}
&~|u^g(x,v)-u^{g^k}(x,v)|\\
=&~\frac{1}{\|v\|}\|u^g(x,v)v-u^{g^k}(x,v)v\|\\
\le&~M_2\|y(t(x,v);x,v)-y^k(t^k(x,v);x,v)\|\\
\le&~M_2[\|y(t^k(x,v);x,v)-y^k(t^k(x,v);x,v)\|\\
&~+\|y(t(x,v);x,v)-y(t^k(x,v);x,v)\|]\\
\le&~M_2\|g^k-g\|M+M_2\left|\int_{t(x,v)}^{t^k(x,v)}\|\dot{y}(t;x,v)\|dt\right|\\
\le&~M_2\|g^k-g\|M+M_2M_3|t(x,v)-t^k(x,v)|\\
\le&~M_2\|g^k-g\|M\\
&~+\frac{2M_2M_3}{M_1}\max\{w^*(x,v)\cdot y(t^k(x,v);x,v),w^*(x,v)\cdot y^k(t(x,v);x,v)\}\\
\le&~\left[M_2M+\frac{2M_2M_3M_4M}{M_1}\right]\|g^k-g\|<\varepsilon
\end{align*}
for all sufficiently large $k$, as desired.

Therefore, $u^{g^k}$ uniformly converges to $u^g$ on any compact set. Choose any $(x,v)\in \mbox{Limsup}_k\succsim^{g^k}$. Then, for any neighborhood $U$ of $(x,v)$ and $k_0$, there exists $k\ge k_0$ such that $\succsim^{g^k}\cap U\neq \emptyset$. This implies that there exists an increasing sequence $k(\ell)$ such that there exists $(x^{\ell},v^{\ell})\in \Omega^2$ such that $u^{g^{k(\ell)}}(x^{\ell},v^{\ell})\ge 1$ and $\|(x^{\ell},v^{\ell})-(x,v)\|<\frac{1}{\ell}$. Because $u^{g^{k(\ell)}}$ uniformly converges to $u^g$ on some compact neighborhood of $(x,v)$, we have that $u^g(x,v)=\lim_{\ell\to\infty}u^{g^{k(\ell)}}(x^{\ell},v^{\ell})\ge 1$. This implies that $(x,v)\in \succsim^g$, and thus $\mbox{Limsup}_k\succsim^{g^k}\subset \succsim^g$. Conversely, suppose that $(x,v)\in \succsim^g$. Then, $u^g(x,v)\ge 1$. Choose any neighborhood $U$ of $(x,v)$. Then, for some $a\in ]0,1[$, $U$ includes $(x,av)$, and $u^g(x,av)>1$. Because $u^{g^k}(x,av)$ converges to $u^g(x,av)$ as $k\to \infty$, there exists $k_0$ such that if $k\ge k_0$, then $u^{g^k}(x,av)\ge 1$, and thus $(x,av)\in \succsim^{g^k}$. This implies that $(x,v)\in \mbox{Liminf}_k\succsim^{g^k}$. Thus, $\succsim^{g^k}$ converges to $\succsim^g$ as $k\to \infty$ with respect to the closed convergence topology. This completes the proof. $\blacksquare$

\section*{References}
\begin{description}
\item{[1]} Antonelli, G. B.: Sulla Teoria Matematica dell' Economia Politica. Tipografia del Folchetto, Pisa (1886). Translated by Chipman, J. S., Kirman, A. P. 1971. On the mathematical theory of political economy. In: Chipman, J. S., Hurwicz, L., Richter, M. K., Sonnenschein, H. F. (Eds.) Preferences, Utility and Demand, pp.333-364. Harcourt Brace Jovanovich, New York (1971)

\item{[2]} Debreu, G.: Representation of a preference ordering by a numerical function. In: Thrall, R. M., Coombs, C. H., Davis, R. L. (Eds.) Decision Processes, pp.159-165. Wiley, New York (1954)

\item{[3]} Debreu, G.: Economies with a finite set of equilibria. Econometrica 38, 387-392 (1970)

\item{[4]} Debreu, G.: Smooth preferences. Econometrica 40, 603-615 (1972)

\item{[5]} Debreu, G.: Smooth preferences, a corrigendum. Econometrica 44, 831-832 (1976)

\item{[6]} Ellsberg, D.: Risk, ambiguity, and the Savage axioms. Q. J. Econ. 75, 643-669 (1961)

\item{[7]} Gale, D.: A note on revealed preference. Economica 27, 348-354 (1960)

\item{[8]} Guillemin, V., Pollack, A.: Differential Topology. Prentice Hall, New Jersey (1974)

\item{[9]} Hartman, P.: Ordinary Differential Equations. second ed. Birkh\"auser Verlag AG, Boston (1997)

\item{[10]} Hosoya, Y.: Measuring utility from demand. J. Math. Econ. 49, 82-96 (2013)

\item{[11]} Hosoya, Y.: A Theory for estimating consumer's preference from demand. Adv. Math. Econ. 18, 33-55 (2015)

\item{[12]} Hosoya, Y.: Recoverability revisited. J. Math. Econ. 90, 31-41 (2020)

\item{[13]} Hosoya, Y.: Consumer optimization and a first-order PDE with a non-smooth system. Oper. Res. Forum 2:66 (2021)

\item{[14]} Hosoya, Y.: Differential characterization of quasi-concave functions without twice differentiability. Commun. Optim. Theory 2022:21, 1-10 (2022)

\item{[15]} Hosoya, Y.: Non-smooth integrability theoty. Forthcoming to Economic Theory (2024)

\item{[16]} Hurwicz, L.: On the problem of integrability of demand functions. In: Chipman, J. S., Hurwicz, L., Richter, M. K., Sonnenschein, H. F. (Eds.) Preferences, Utility and Demand, pp.174-214. Harcourt Brace Jovanovich, New York (1971)

\item{[17]} Hurwicz, L., Richter, M. K.: An integrability condition with applications to utility theory and thermodynamics. J. Math. Econ. 6, 7-14 (1979)

\item{[18]} Hurwicz, L., Uzawa, H.: On the integrability of demand functions. In: Chipman, J. S., Hurwicz, L., Richter, M. K., Sonnenschein, H. F. (Eds.) Preferences, Utility and Demand, pp.114-148. Harcourt Brace Jovanovich, New York (1971)

\item{[19]} Iserles, A.: A First Course in the Numerical Analysis of Differential Equations, 2nd ed. Cambridge University Press, Cambridge (2009)

\item{[20]} Katzner, D. W.: A note on the differentiability of consumer demand functions. Econometrica 36, 415-418 (1968)

\item{[21]} Kawamata, M.: On the role of Wieser's theory of natural value played in the history of marginal utility theory. Keio J. Econ. 82-2, 275-296.

\item{[22]} Kihlstrom, R., Mas-Colell, A., Sonnenschein, H.: The demand theory of the weak axiom of revealed preference. Econometrica 44, 971-978 (1976)

\item{[23]} Mas-Colell, A.: The recoverability of consumers' preferences from market demand behavior. Econometrica 45, 1409-1430 (1977)

\item{[24]} Menger, C. Grunds\"atze der Vorkswirthschaftslehre. Wilhelm Braum\"uller, Wien (1871)

\item{[25]} Otani, K.: A characterization of quasi-concave functions. J. Econ. Theory 31, 194-196 (1983)

\item{[26]} Pareto, V.: Manuale di Economia Politica con una Introduzione alla Scienza Sociale. Societa Editrice Libraria, Milano (1906)

\item{[27]} Pareto, V.: Manuel d'Economie Politique. Paris: Giard et E. Briere (1909)

\item{[28]} Richter, M. K.: Revealed preference theory. Econometrica 34, 635-645 (1966)

\item{[29]} Samuelson, P. A.: The problem of integrability in utility theory. Economica 17, 355-385 (1950)

\item{[30]} Schmeidler, D.: Subjective probability and expected utility without additivity.'' Econometrica 57, 571-587 (1989)

\item{[31]} Suda, S.: Vilfredo Pareto and the Integrability Problem of Demand Functions. Keio J. Econ. 99-4. 637-655 (2007)

\item{[32]} Ville, J.: Sur les conditions d'existence d'une oph\'elimit\'e totale et d'un indice du niveau des prix. Annales de l'Univers\`it\'e de Lyon 8, Sec. A(3), 32-39 (1946)

\item{[33]} Volterra, V.: L'economia matematica ed il nuovo manuale del prof. Pareto. Giornale degli Economisti 32, 296-301 (1906)

\end{description}

\end{document}